\documentclass[prl,amsmath,amssymb, twocolumn,superscriptaddress, 11pt,tightenlines,aps]{revtex4-1}
\usepackage{graphicx}
\usepackage{dcolumn}   
\usepackage{bm}        
\usepackage{bbm}
\usepackage{dsfont}
\usepackage{mathtools}
\usepackage{amsthm, thmtools}
\usepackage{amsmath}
\usepackage{empheq}
\usepackage{enumerate}

\newcommand{\tinyspace}{\mspace{1mu}}
\newcommand{\microspace}{\mspace{0.5mu}}

\newcommand{\norm}[1]{\left\lVert\tinyspace#1\tinyspace\right\rVert}

\declaretheorem[
]{proposition}

\def\<{\langle}
\def\>{\rangle}
\def \lket {\left|}
\def \rket {\right\rangle}
\def \lbra {\left\langle}
\def \rbra {\right|}
\newcommand{\ket}[1]{\lket\microspace #1 \microspace\rket}
\newcommand{\oket}[1]{\lket\microspace #1 \microspace\right )}
\newcommand{\bra}[1]{\lbra\microspace #1 \microspace\rbra}
\newcommand{\obra}[1]{\left (\microspace #1 \microspace\rbra}

\newcommand{\Tr}{\mathrm{Tr}}

\newcommand{\abs}[1]{\left | #1 \right|}

\newcommand{\sld}{\text{\tiny{\textup{SLD}}}}
\newcommand{\rld}{\text{\tiny{\textup{RLD}}}}

\usepackage[usenames,dvipsnames]{xcolor}

\usepackage{bm}

\usepackage[colorlinks=true,citecolor=Cerulean,linkcolor=RubineRed,urlcolor=Cerulean]{hyperref}
\usepackage[capitalize]{cleveref}

\begin{document}
\title{Generalized geometric speed limits for quantum observables}
\author{Jacob Bringewatt}
\affiliation{Department of Physics, Harvard University, Cambridge, MA 02138 USA}
\affiliation{Joint Center for Quantum Information and Computer Science, NIST/University of Maryland, College Park, Maryland 20742, USA}
\affiliation{Joint Quantum Institute, NIST/University of Maryland, College Park, Maryland 20742, USA}
\author{Zach Steffen}
\affiliation{Maryland Quantum Materials Center, University of Maryland, College Park, MD 20742, USA}
\affiliation{Laboratory for Physical Sciences, 8050 Greenmead Drive, College Park, MD 20740, USA}
\author{Martin A. Ritter}
\affiliation{Joint Quantum Institute, NIST/University of Maryland, College Park, Maryland 20742, USA}
\author{Adam Ehrenberg}
\affiliation{Joint Center for Quantum Information and Computer Science, NIST/University of Maryland, College Park, Maryland 20742, USA}
\affiliation{Joint Quantum Institute, NIST/University of Maryland, College Park, Maryland 20742, USA}
\author{Haozhi Wang}
\affiliation{Maryland Quantum Materials Center, University of Maryland, College Park, MD 20742, USA}
\affiliation{Laboratory for Physical Sciences, 8050 Greenmead Drive, College Park, MD 20740, USA}
\author{B. S. Palmer}
\affiliation{Maryland Quantum Materials Center, University of Maryland, College Park, MD 20742, USA}
\affiliation{Laboratory for Physical Sciences, 8050 Greenmead Drive, College Park, MD 20740, USA}
\author{Alicia J. Koll\'ar}
\affiliation{Joint Quantum Institute, NIST/University of Maryland, College Park, Maryland 20742, USA}
\author{Alexey V. Gorshkov}
\affiliation{Joint Center for Quantum Information and Computer Science, NIST/University of Maryland, College Park, Maryland 20742, USA}
\affiliation{Joint Quantum Institute, NIST/University of Maryland, College Park, Maryland 20742, USA}
\author{Luis Pedro Garc\'ia-Pintos}
\affiliation{Theoretical Division (T-4), Los Alamos National Laboratory, New Mexico, 87545, USA}

\date{\today}
\begin{abstract}
Leveraging quantum information geometry, we derive generalized quantum speed limits on the rate of change of the expectation values of observables. These bounds subsume and, for Hilbert space dimension $\geq 3$, tighten existing bounds---in some cases by an arbitrarily large multiplicative constant. The generalized bounds can be used to design ``fast'' Hamiltonians that enable the rapid driving of the expectation values of observables with potential applications e.g.~to quantum annealing, optimal control,  variational quantum algorithms, and quantum sensing. Our theoretical results are supported by illustrative examples and an experimental demonstration using a superconducting qutrit. Possibly of independent interest, along the way to one of our bounds we derive a novel upper bound on the generalized quantum Fisher information with respect to time (including the standard symmetric logarithmic derivative quantum Fisher information) for unitary dynamics in terms of the variance of the associated Hamiltonian and the condition number of the density matrix.
\end{abstract}
\maketitle

\textit{Introduction.---}Energy-time uncertainty relations and the associated quantum speed limits were first formalized by Mandelstam and Tamm~\cite{mandelstam1945uncertainty}. 
They provided a lower bound for the time $t^\perp$ for a quantum system in a pure state to reach an orthogonal state, given unitary evolution under some Hamiltonian.
Since then, a variety of other bounds, all under the general heading of ``quantum speed limits,'' have been derived, providing bounds on the rate of change of quantum states~\cite{margolus1998maximum,PhysRevLett.110.050403,
PhysRevLett.110.050402,
PhysRevLett.111.010402,
pires2016generalized,campaioli2018tightening,campaioli2019tight,oconnor2020action,vanvu2023topological,mai2023tight} and observables~\cite{lieb1972finite,lgp2022unifying,Carabba2022quantumspeedlimits,mohan2022quantum,H_rnedal_2022}, with applications to metrology~\cite{maleki2023speed,herb2024quantum}, quantum thermodynamics~\cite{campbell2018precision,campaioli2017enhancing,
PhysRevLett.125.040601,funo2019speed,das2021thermo,vanvu2023thermo}, quantum control theory~\cite{caneva2009optimal}, the analysis of quantum algorithms~\cite{lgp2023lower} and many others~\cite{deffner2017quantum}.

Of particular note, Ref.~\cite{lgp2022unifying} derived a set of information-theoretic speed limits on the rate of change of expectation values of observables that hold for any probability-conserving dynamics.
For a fixed observable $A$ and a state $\rho$, 
\begin{equation}\label{eq:lgpbound}
\abs{\dot a}\leq \Delta A \sqrt{\mathcal{I}^\sld},
\end{equation}
where $\dot a:=\Tr\left[A \dot\rho\right]\,$ (here, $\dot\rho:=d\rho/dt$) is the velocity of the expectation value of $A$ in the state $\rho$, $(\Delta A)^2:=\mathrm{Tr}(\rho A^2)-\mathrm{Tr}(\rho A)^2$ is the variance of $A$, and $\mathcal{I}^\sld$ is the so-called symmetric logarithmic derivative (SLD) quantum Fisher information (QFI) with respect to time. 
The SLD QFI has a simple interpretation: it quantifies how much the state $\rho$ changes as 
$t\rightarrow t + dt$.

\Cref{eq:lgpbound} consists of two components: (1) a dynamics-independent term, $\Delta A$, that depends on the uncertainty of the observable in question with respect to the current state $\rho$; (2) an observable-independent term, $\sqrt{\mathcal{I}^\sld}$ that depends only on the underlying state and the dynamics of the system.
Both of these terms have geometric content, as they can each be associated with a natural notion of an inner product between Hermitian operators. In particular, considering $\rho$ as a point in a one-dimensional manifold of quantum states parameterized by the time $t$, 
$\mathcal{I}^\sld$  is a natural choice of
Riemannian metric on the space of density matrices
~\cite{petz1996monotone}. It is a metric in
that it provides a notion of distance between each point $\rho(t)$ and $\rho(t+dt)$.

Stronger versions of \cref{eq:lgpbound} can be obtained by splitting $A$ and $\mathcal{I}^{\sld}$ into their incoherent (diagonal) and coherent components (off-diagonal) in the eigenbasis of $\rho$~\cite{lgp2022unifying}. In particular,
\begin{equation}\label{eq:lgp-upper}
|\dot a|\leq (\Delta A_C)\sqrt{\mathcal{I}_C^\sld}+(\Delta A_I) \sqrt{\mathcal{I}_I^\sld},
\end{equation}
where, the subscripts $C$ and $I$ denote the coherent and incoherent parts, respectively. We have not yet defined precisely what the incoherent and coherent parts of $\mathcal{I}^\sld$ are ($\mathcal{I}^\sld$ is a number, not an operator!), but, qualitatively, $\mathcal{I}^\sld_C$ ($\mathcal{I}^\sld_I$) quantifies 
how much the eigenvectors (eigenvalues) of $\rho$ change as $t\rightarrow t+dt$.

\textit{Key Results.---}Before delving into the mathematical details, we summarize the key results:
First, we show that the bounds in \cref{eq:lgpbound} and \cref{eq:lgp-upper} are each but one example in a much richer family of geometric quantum speed limits. In particular, we leverage the fact that the choice of inner product on the (tangent space of the) manifold of quantum states that leads to $\mathcal{I}^\sld$ and $\Delta A$ is not unique~\cite{petz1996monotone,petz2002covariance}; thus, both the SLD QFI and the variance can be replaced with generalized quantities in \cref{eq:lgpbound,eq:lgp-upper}. 
Subject to the natural constraint that the associated Reimannian metric (i.e.~QFI) is contractive (i.e.~distances shrink under noise), these generalized notions of QFI and of variance are characterized by a special set of 
functions $f$ (defined rigorously below). We demonstrate that, for any valid choice of $f$,
\begin{align}
&|\dot a|
\leq (\Delta^f A)\sqrt{\mathcal{I}^f},\label{eq:new-bnd}
\end{align}
and, also, 
\begin{equation}\label{eq:new-upper}
|\dot a|\leq (\Delta^f A_C)\sqrt{\mathcal{I}_C^f}+(\Delta A_I) \sqrt{\mathcal{I}_I}.
\end{equation}
We dropped the $f$-superscript on the incoherent terms because, for incoherent operators, the generalized QFI and variance are identical for all $f$. 
We prove that \cref{eq:new-upper} is 
tighter than \cref{eq:new-bnd}. 

\cref{eq:lgpbound,eq:lgp-upper} correspond to one particular choice of $f$.
For Hilbert space dimension $\geq 3$, we show that, for the optimal choice of $f$, \cref{eq:new-bnd,eq:new-upper} are generically tighter than \cref{eq:lgpbound,eq:lgp-upper}, respectively,
in some cases by an arbitrarily large multiplicative factor. These improvements come from taking advantage of the freedom in choosing a natural inner product on the manifold of quantum states to tighten the slack in the Cauchy-Schwarz inequality that underlies such bounds. Our results are supported by a toy example and an experimental demonstration using a superconducting qutrit.
We also show how our bounds inform the construction of novel control Hamiltonians to rapidly drive observables with applications to, e.g., quantum annealing and quantum machine learning. 

Finally, for arbitrary coherent dynamics driven by a Hamiltonian $H$
and assuming incoherent dynamics driven by entanglement with an environment via a Hamiltonian $H^\mathrm{int}$, we obtain a looser, but simpler, speed limit
\begin{equation}\label{eq:variance-bnd}
|\dot a|\leq \sqrt{\kappa_\rho} \Delta^f A_C\Delta H+2\Delta A_I\Delta H^\mathrm{int},
\end{equation}
where $\kappa_\rho$ is the condition number of $\rho$ (i.e. the ratio of its largest eigenvalue to its smallest eigenvalue). Importantly, unlike $\mathcal{I}^f$, this bound does not depend directly on $\dot\rho$ and is expressed in terms of more easily accessible physical quantities. For $f = f_\sld$ and mixed states with $\kappa_\rho < 4$, \cref{eq:variance-bnd} tightens the original speed limit on observables derived by Mandelstam and Tamm~\cite{mandelstam1945uncertainty}. Along the way to \cref{eq:variance-bnd}, we derive a new bound on any QFI that
may be of independent interest.

\textit{Mathematical Details.---}Let $\mathcal{M}_d$ denote the set of $d$-dimensional Hermitian matrices. Define the space of density operators $\mathcal{D}:=\{\rho\in\mathcal{M}_d\,|\,\Tr\,\rho=1, \rho\geq 0\}$ over the $d$-dimensional Hilbert space $\mathcal{H}_d$. Restricting our attention to positive-definite $\rho\in\mathcal{D}$,
we have  
a Riemannian manifold~\footnote{In this paper, we only consider positive-definite states. A subset of the metrics considered here admit an extension to pure states, but, here, they all correspond to the Fubini-Study metric~\cite{petz1996geometries}, so generalized geometric speed limits for pure states are uninteresting. There is a large family of non-positive definite, non-pure states, but, to our knowledge, metrics on this full topological boundary have not been considered in the literature.}, i.e.~we can associate a Riemannian metric $g_\rho$ with each point $\rho$ on the manifold of positive-definite states, providing a notion of infinitesimal distance between two full rank states $\rho$ and $\rho+d\rho$. 

It is natural to require that $g$ be contractive under noise, 
i.e.~$g_{T\rho}(TA, TB)\leq g_\rho(A,B)$, where $T$ is a completely positive trace-preserving map and $A,B$ are elements of the tangent space at $\rho$ (traceless Hermitian matrices of dimension $d$)~\footnote{While monotonicity is a natural requirement for a metric, one can define geometric speed limits without this requirement~\cite{campaioli2019tight}.}. 
As alluded to above, all such metrics can be identified with a function $f:\mathbb{R}^+\rightarrow\mathbb{R}^+$ known as a symmetric, normalized operator monotone function~\cite{petz1996monotone}.
These functions $f$ obey:
(i) continuity; (ii) for matrices $A\geq B> 0$, $f(A)\geq f(B)> 0$ (operator monotonicity); (iii) $x f(x^{-1})=f(x)$, for $x> 0$ (symmetry); (iv) $f(1)=1$ (normalization). Call the set of functions satisfying these properties $\mathbb{F}$.

These functions are important because they lead to a generalized mean $m^f(x, y):=xf(x^{-1}y)$ of a pair of either positive real numbers or positive definite, commuting (super-) operators~\footnote{The definition can be adapted slightly to allow for non-commuting operators, but as we do not need this generalization we keep the mean in its simplest form.}~\cite{kubo1980means}.
The properties of $\mathbb{F}$ imply various natural properties for the mean $m^f(x,y)$, e.g., $m^f(x,x)=x$ and $m^f(x,y)=m^f(y,x)$~\cite{SM}.

In the context of quantum speed limits, consider $A=B=\dot\rho$; the distance $\mathcal{I}^f=4g_\rho^f\big(\dot\rho,\dot\rho\big)$ is called the generalized QFI with respect to time for the state $\rho$~\footnote{The normalization factor of $1/4$ in the definition of $g_\rho$ is the conventional choice as it makes $g_\rho^f$ coincide with the Euclidean metric restricted to the positive orthant of the unit sphere defined in the space of eigenvalues of $\rho$.}. In particular~\cite{petz1996geometries,gibilisco2024f},
\begin{align}\label{eq:gen-qfi}
\mathcal{I}^f&:=
\mathrm{Tr}\left[\dot\rho\, m^f(L_\rho,R_\rho)^{-1}(\dot\rho)\right],
\end{align}
where $L_\rho(\cdot):=\rho(\cdot)$, $R_\rho(\cdot):=(\cdot)\rho$, and $m^f(L_\rho,R_\rho)$ are positive definite super-operators and, thus, invertible.

We will consider a particular one-parameter set of functions $\mathbb{F}_\beta\subset \mathbb{F}$~\cite{gibilisco2004characterisation}. Two important elements are
$f_{\beta=1}(x)=(1+x)/2$, corresponding to the standard SLD QFI, and
$f_{\beta=-1}(x)=2x/(1+x)$,
corresponding to the right logarithmic derivative (RLD) QFI. 
The associated mean $m^f$ is largest (smallest) over all $\mathbb{F}$ for $f_\sld$ ($f_\rld$), where it corresponds to the arithmetic (harmonic) mean. See the supplemental material for a full mathematical description of $\mathbb{F}_\beta$~\cite{SM}.

We can also define a generalized variance as
\begin{equation}\label{eq:gen-var}
(\Delta^fA)^2:=\Tr\left[A_0\, m^f(L_\rho, R_\rho) (A_0)\right],
\end{equation}
where $A$ is Hermitian and $A_0:=A-\Tr(\rho A)$~\cite{petz2002covariance}. For $f=f_\sld$, \cref{eq:gen-var} is the usual variance.

Using these definitions, it is straightforward to derive \cref{eq:new-bnd,eq:new-upper}. To begin, define an $f$-dependent ``logarithmic derivative'' operator $
\mathcal{L}^f:=m^f(L_\rho,R_\rho)^{-1}(\dot\rho)$. Note that $\mathcal{I}^f=(\Delta\mathcal{L}^f)^2$. Using that $L_\rho$ and $R_\rho$ are positive definite, $|\dot a|=\big|\Tr[\dot\rho A]\big|=\big|\Tr[\dot\rho A_0]\big|=\big|\Tr[m^f(L_\rho,R_\rho)(\mathcal{L}^f)A_0]\big|$. Applying
the Cauchy-Schwarz inequality gives \cref{eq:new-bnd}.

Splitting $A$ and $\mathcal{L}^f$ into coherent and incoherent parts in the eigenbasis of $\rho=\sum_j p_j \ket{j}\bra{j}$, a similar derivation~\cite{SM} yields the bound in \cref{eq:new-upper}.
Furthermore, 
one finds that \cref{eq:new-upper} is generically tighter than \cref{eq:new-bnd}, coinciding if and only if $(\Delta^fA_C)\sqrt{\mathcal{I}_I}=(\Delta A_I)\sqrt{\mathcal{I}^f_C}$~\cite{SM}.

\textit{Bound Comparison.---}While the new speed limits specified by \cref{eq:new-bnd,eq:new-upper} are of independent interest from the perspective of information geometry, to be of broader utility,  there must exist triplets $(\rho,\dot\rho,A)$ where the tightest bound corresponds to $f\neq f_\sld$,  so that the new bounds are tighter than those of Ref.~\cite{lgp2022unifying}. It is not immediately clear if this is possible: on one hand, $\mathcal{I}^f\geq \mathcal{I}^\sld$ for all $f\in\mathbb{F}$; on the other,
$\Delta^f A\leq \Delta^\sld A$ for all $f\in\mathbb{F}$.
Tighter bounds depend on the tightening of the variance term winning out over the loosening of the QFI term.

As the incoherent term in 
\cref{eq:new-upper}
is identical for all $f$, to find improved bounds it is sufficient to consider only the coherent term.
To this end, define the coherent ratio
$\xi^f:=\frac{(\Delta^fA_C)^2\mathcal{I}^f_C}{(\Delta^\sld A_C)^2\mathcal{I}^\sld_C}.$
If $\xi^f<1$, the generalized bound \cref{eq:new-upper} is tighter than the case with $f=f_\sld$ derived in Ref.~\cite{lgp2022unifying}. 

For a qubit, it is straightforward to show that
no improvement is possible from the generalized bounds~\cite{SM}. However, moving to larger systems, a qutrit ($d=3$) is already sufficient to demonstrate two facts: (1) there are triplets $(\rho, \dot\rho, A)$ such that the generalized bounds with $f\neq f_\sld$ are tightest; (2) for every $f_{\beta}\in\mathbb{F}_\beta$, there exists $(\rho,\dot\rho,A)$ such that the tightest bound, optimized over this family, is $f_{\beta}$.

\textit{Example.} Consider a qutrit. In the eigenbasis of $\rho=\sum_j p_j \ket{j}\bra{j}$, let 
$A_{02}=A_{12}=0$,
$v_{ij}:=2|(\dot \rho)_{ij}|^2$, and $m^f_{ij}:=m^f(p_i,p_j)$. 
Then~\cite{SM}
\begin{equation}\label{eq:qutrit-comp}
\xi^f=\frac{v_{01}+v_{02}\Big(\frac{m^f_{01}}{m^f_{02}}\Big)+v_{12}\Big(\frac{m^f_{01}}{m^f_{12}}\Big)}{v_{01}+v_{02}\Big(\frac{m^\sld_{01}}{m^\sld_{02}}\Big)+v_{12}\Big(\frac{m^\sld_{01}}{m^\sld_{12}}\Big)}.
\end{equation}

Consider evolution under a Hamiltonian $H=i(\Omega/2)(\ket{g}\bra{e}+\ket{e}\bra{f})+h.c.$, 
where $\ket{g},\ket{e},\ket{f}$ are 
the computational basis states. For simplicity, consider $\rho$ diagonal in the computational basis, so that we can identify the computational basis with the eigenbasis 
of $\rho$. 
Thus, $v_{02}=0$, 
so we can ignore the middle term in both the numerator and denominator of \cref{eq:qutrit-comp} and turn our attention to the third terms. Using that $m^\sld_{ij}\geq m^f_{ij}$ $\forall$ $f$, $\xi^f<1$ if $p_1\approx p_2$ (so that $m^\sld_{12}\approx m^f_{12}$, as $m^f(x,x+\epsilon)\approx x$ for all $f$ if $\epsilon\ll 1$); and, also, $|p_0-p_1|$ is sufficiently large (so that $m^\sld_{01}\gg m^f_{01}$).

This qualitative description is made precise in \cref{fig:example-1}, where we minimize $\xi^f$ over $f_\beta\in\mathbb{F}_\beta$ for equally-spaced 
grid points $(p_0,p_1,p_2)$ on the probability simplex 
with $p_j\in[0.0025,0.995]$ $\forall\, j$~\cite{code}. There are large regions of parameter space where the generalized bounds are tighter than existing bounds. Here, the optimal $f_\beta$ is $f_{-1}=f_\rld$. 
For $p_0=\epsilon^3, p_1=\epsilon, p_2=1-p_0-p_1$, we have $\xi^f = \mathcal{O}(\epsilon)$, 
indicating that arbitrarily large multiplicative improvements over existing bounds are possible by taking $\epsilon \rightarrow 0$~\cite{SM}.

This example is also well-suited for a proof-of-principle experimental demonstration. Leaving the details to the supplemental material~\cite{SM},  we prepare a superconducting qutrit in a diagonal mixed state $\rho$ by letting the qutrit naturally decay for a time $t_{\mathrm{decay}}$ from the $\ket{f}$ state.
Driving both the $\ket{g}\leftrightarrow\ket{e}$ and $\ket{e}\leftrightarrow\ket{f}$ transitions, 
we extract the speed $|\dot a|$
of the Pauli-$X$ observable in the $\{\ket{g}, \ket{e}\}$ subspace from estimates of the rate of change of the populations of these states. 
We repeat this for many initial states prepared via the natural decay of the qutrit and compare the measured speeds to the  $f_\sld-$ and $f_\rld-$based bounds, computed directly from the experimentally-determined states at 10$\mu$s intervals along the decay trajectory,
as shown in the inset of \cref{fig:example-1}. While both bounds constrain the measured speeds within experimental error for small $t_\mathrm{decay}$, the new $f_\rld$-based bound is significantly tighter than the $f_\sld$-based bound. 

The problem becomes richer if $v_{01}, v_{02}, v_{12}\neq 0$. 
Here, the optimal choice of $f_\beta\in\mathbb{F}_\beta$ can occur for any $\beta\in[-1,1]$ as one moves through the parameter space, demonstrating fact (2) above~\cite{SM}.

\begin{figure}[t]
\centering
\includegraphics[width=\columnwidth]{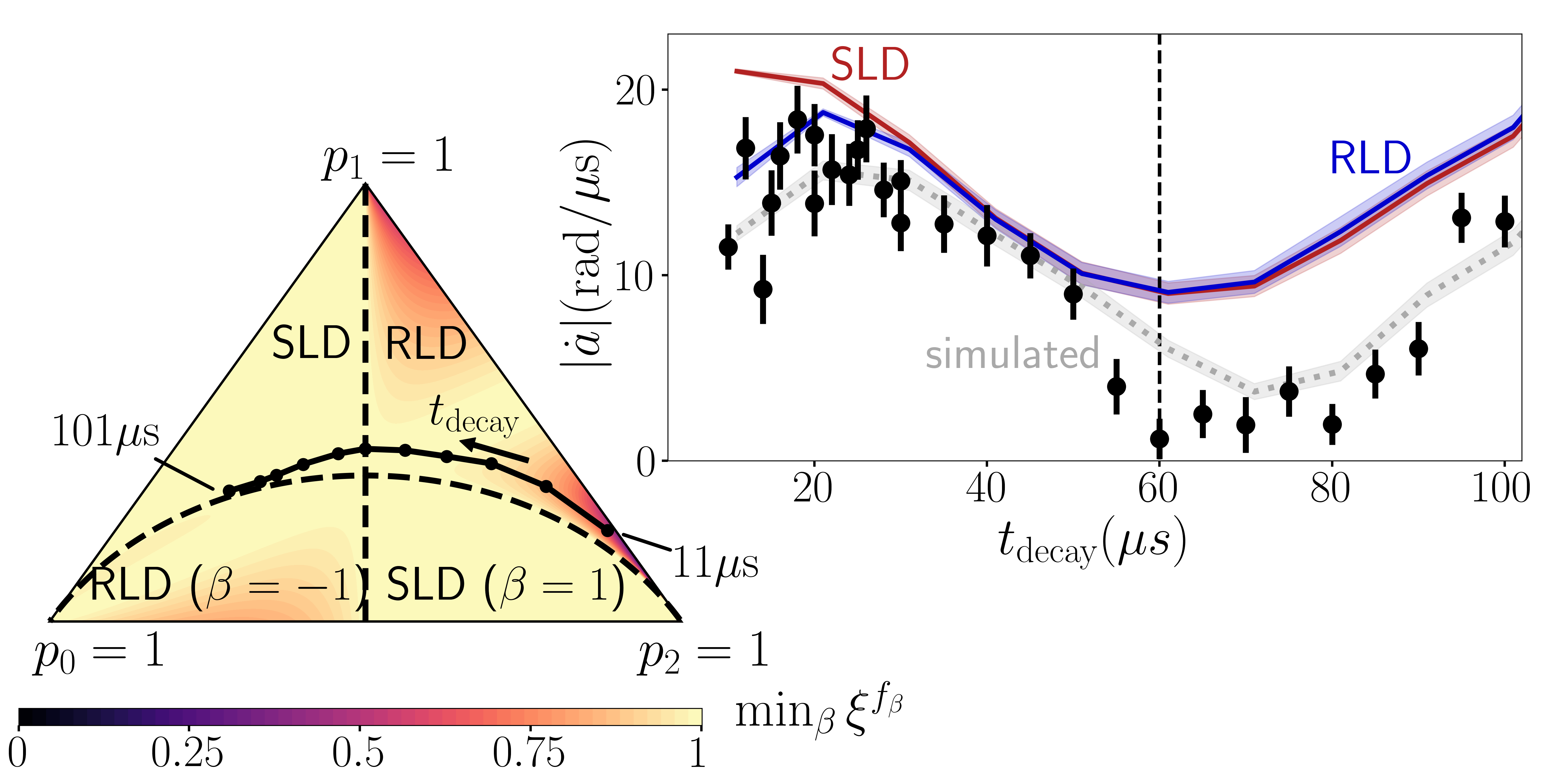}
\caption{For the qutrit example described in the text,
$\xi^f$ for the optimum $f_\beta\in\mathbb{F}_\beta$ as a function of the eigenvalues $p_0, p_1,p_2$ of the diagonal state $\rho$, plotted using barycentric coordinates. $\xi^f < 1$ corresponds to the new bounds being tighter. In these regions, $\xi^f$ is minimized for $f_{\beta = -1}=f_\rld$. The dashed lines in the main plot and the inset separate these regions from where there is no improvement over the SLD-based bounds. The states used in the experimental demonstration are prepared by letting the qutrit decay from $\ket{f}$ for a time $t_\mathrm{decay}\in [11,101]\, \mu s$ leading to the one-parameter family of states depicted with the black curve on the probability simplex. 
In the inset, the experimentally extracted value for $|\dot a|$ driving at a Rabi frequency of $\Omega=(2\pi) 10$ MHz for a number of initial states along this trajectory is compared to the RLD and SLD bounds computed directly from experimentally-determined states. The dotted gray line shows the simulated speed for these same states. Error bands are the standard deviation of the mean over 21 preparations of the states. For small $t_\mathrm{decay}$, the new $f_\rld$-based bound is significantly tighter than the old $f_\sld$-based bound.
}\label{fig:example-1}
\end{figure}

{\textit{Saturation and ``Fast'' Hamiltonians.---}}As the upper bounds come from
the Cauchy-Schwarz inequality, \cref{eq:new-bnd,eq:new-upper} are tight when $A_0\propto \mathcal{L}^f$. 
In the eigenbasis of $\rho$, this corresponds to the condition $(A_0)_{ij}=\gamma\frac{\dot\rho_{ij}}{m^f(p_i,p_j)}$ for some constant $\gamma$. 
In the example above, $A$ has only a single pair of non-zero components $A_{ij}, A_{ji}$; thus if the bounds are tight for any $f$, they are tight for all $f$.
In general, however, the saturation of the speed limits is $f$-dependent. In fact, one can construct examples where the SLD-based bound is loose, while, with an appropriate choice of $f$, the generalized bounds are tight~\cite{SM}.

A natural question immediately follows: 
what choice of dynamics leading to the saturation of \cref{eq:new-upper} (determined by a choice of $f\in\mathbb{F}$) maximizes $|\dot a|$? 
As the $f$-dependence of the answer relies only on the coherent part of \cref{eq:new-upper}, we restrict our attention to Hamiltonian dynamics to emphasize the role of $f$. For the appropriate choice of  Hamiltonian $H^{f,\mathrm{fast}}_{\rho,A}$, it holds that $|\dot a|=(\Delta^f A_C)\sqrt{\mathcal{I}^f_C}$. 
Using that $\dot\rho=-i[H,\rho]$, 
\begin{align}\label{eq:Hfast}
H^{f,\mathrm{fast}}_{\rho,A}=-\frac{i}{\gamma} \sum_{j\neq k}\frac{m^f(p_j,p_k)}{p_j-p_k}A_{jk}\ket{j}\!\bra{k}
\end{align}
for some constant $\gamma$~\footnote{Note that this gives us flexibility to fix the Hamiltonian norm.}, where we assume that $p_j\neq p_k$. 
The optimal $f$ that maximizes $|\dot a|$ over coherent-quantum-speed-limit-saturating dynamics will be state- and observable dependent. 
Importantly, using $H_{\rho, A}^{f, \mathrm{fast}}$ for $f\neq f_\sld$ can lead to larger $|\dot a|$ than $f_\sld$ alone. See supplemental material for an example~\cite{SM}.

Such Hamiltonians may be hard to implement for large systems, but they could guide the construction of alternative Hamiltonians built from a limited set of controls. This approach could help to design~\cite{magann2022feedback} and understand the performance~\cite{lgp2023lower} of quantum annealing algorithms or to provide improved versions of quantum natural gradient descent~\cite{stokes2020quantum} or, for quantum sensing, to determine optimal Hamiltonians with which to couple an unknown parameter, given a limited set of observables.

\textit{Bounds Using Energy Variance.---}It is desirable to also derive quantum speed limits in terms of more directly physical and easily accessible quantities than the generalized QFIs. Consider unitary dynamics driven via a Hamiltonian $H$ and assume any nonunitary dynamics are due to entanglement with an environment via a Hamiltonian $H^\mathrm{int}$ consisting of all terms with support on both the system and the environment. Then 
\begin{align}\label{eq:I-bnds}
\mathcal{I}_C^f\leq \kappa_\rho (\Delta H)^2, && \mathcal{I}_I\leq 4 (\Delta H^\mathrm{int})^2,
\end{align}
where $\kappa_\rho$ is the condition number of $\rho$ and $(\Delta H^\mathrm{int})^2$ is computed for the joint state of the system and the environment. The first bound is new~\cite{SM}; the second is from Ref.~\cite{lgp2022unifying}. Plugging \cref{eq:I-bnds} into \cref{eq:new-upper} yields \cref{eq:variance-bnd}.

As $\Delta^f A \leq \Delta A$ for all $f$, the tightest bound in \cref{eq:variance-bnd} will always correspond to $f=f_\rld$. However, even for the loosest choice $f=f_\sld$, if $\kappa_\rho<4$ then \cref{eq:variance-bnd} is tighter than existing similar bounds, given by \cref{eq:variance-bnd} with $\sqrt{\kappa_\rho}(\Delta^f A_C)\rightarrow 2(\Delta A_C)$~\cite{lgp2022unifying}. $\kappa_\rho<4$ holds for nearly fully mixed states; for instance, a thermal state $\rho\propto \exp(-\beta H)$ with $\beta \leq \log(4)/(2\norm{H}_s)$.  
See supplemental material for more discussion~\cite{SM}. 

\textit{Discussion.---}Leveraging the tools of quantum information geometry, we prove novel quantum speed limits for observables that are 
generically tighter than previously derived bounds~\cite{lgp2022unifying}.
The family of inner products we use 
have also been used to derive tighter bounds for multiparameter quantum metrology~\cite{lu2020generalized}, generalized uncertainty relations~\cite{gibilisco2008robertson,yanagi2011metric}, and quantum speed limits on the evolution of states~\cite{pires2016generalized,pires2024experimental}. 
We expect these inner products and the novel bounds in this work
can also be applied  
in other settings where the SLD QFI is used~\cite{meyer2021fisher}, e.g.,~for analyzing quantum machine learning models~\cite{abbas2021power,haug2021capacity}
or parameter estimation problems~\cite{garcia2024estimation}.

{\textit{Acknowledgements.---}}We thank Sean Muleady and Yu-Xin Wang for helpful conversations.
We appreciate support from a seed grant via the NSF Quantum Leap Challenge (QLCI) for Robust Quantum Simulation (RQS) (award No.~OMA-2120757).
J.B., A.E., and A.V.G.~were also supported in part by the DoE ASCR Quantum Testbed Pathfinder program (awards No.~DE-SC0019040 and No.~DE-SC0024220), DoE ASCR Accelerated Research in Quantum Computing program (award No.~DE-SC0020312), NSF STAQ program, AFOSR, AFOSR MURI, and DARPA SAVaNT ADVENT. Support is also acknowledged from the U.S.~Department of Energy, Office of Science, National Quantum Information Science Research Centers, Quantum Systems Accelerator. 
M.R. received support from the National Science Foundation (PFC at JQI Grant No.  PHY-1430094), the Laboratory for Physical Sciences Quantum Graduate Fellowship, and ARCS. 
L.P.G.P. acknowledges support by the DOE Office of Science, Office of Advanced Scientific Computing Research, Accelerated Research for Quantum Computing program, Fundamental Algorithmic Research for Quantum Computing (FAR-QC) project, Laboratory Directed Research and Development (LDRD) program of LANL under project number 20230049DR, and Beyond Moore’s Law project of the Advanced Simulation and Computing Program at LANL managed by Triad National Security, LLC, for the National Nuclear Security Administration of the U.S. DOE under contract 89233218CNA000001.

\bibliography{main}

\begin{thebibliography}{59}%
\makeatletter
\providecommand \@ifxundefined [1]{%
 \@ifx{#1\undefined}
}%
\providecommand \@ifnum [1]{%
 \ifnum #1\expandafter \@firstoftwo
 \else \expandafter \@secondoftwo
 \fi
}%
\providecommand \@ifx [1]{%
 \ifx #1\expandafter \@firstoftwo
 \else \expandafter \@secondoftwo
 \fi
}%
\providecommand \natexlab [1]{#1}%
\providecommand \enquote  [1]{``#1''}%
\providecommand \bibnamefont  [1]{#1}%
\providecommand \bibfnamefont [1]{#1}%
\providecommand \citenamefont [1]{#1}%
\providecommand \href@noop [0]{\@secondoftwo}%
\providecommand \href [0]{\begingroup \@sanitize@url \@href}%
\providecommand \@href[1]{\@@startlink{#1}\@@href}%
\providecommand \@@href[1]{\endgroup#1\@@endlink}%
\providecommand \@sanitize@url [0]{\catcode `\\12\catcode `\$12\catcode `\&12\catcode `\#12\catcode `\^12\catcode `\_12\catcode `\%12\relax}%
\providecommand \@@startlink[1]{}%
\providecommand \@@endlink[0]{}%
\providecommand \url  [0]{\begingroup\@sanitize@url \@url }%
\providecommand \@url [1]{\endgroup\@href {#1}{\urlprefix }}%
\providecommand \urlprefix  [0]{URL }%
\providecommand \Eprint [0]{\href }%
\providecommand \doibase [0]{http://dx.doi.org/}%
\providecommand \selectlanguage [0]{\@gobble}%
\providecommand \bibinfo  [0]{\@secondoftwo}%
\providecommand \bibfield  [0]{\@secondoftwo}%
\providecommand \translation [1]{[#1]}%
\providecommand \BibitemOpen [0]{}%
\providecommand \bibitemStop [0]{}%
\providecommand \bibitemNoStop [0]{.\EOS\space}%
\providecommand \EOS [0]{\spacefactor3000\relax}%
\providecommand \BibitemShut  [1]{\csname bibitem#1\endcsname}%
\let\auto@bib@innerbib\@empty
\bibitem [{\citenamefont {Mandelstam}\ and\ \citenamefont {Tamm}(1945)}]{mandelstam1945uncertainty}%
  \BibitemOpen
  \bibfield  {author} {\bibinfo {author} {\bibfnamefont {L.}~\bibnamefont {Mandelstam}}\ and\ \bibinfo {author} {\bibfnamefont {I.}~\bibnamefont {Tamm}},\ }\href@noop {} {\bibfield  {journal} {\bibinfo  {journal} {J. Phys.(USSR)}\ }\textbf {\bibinfo {volume} {9}},\ \bibinfo {pages} {249} (\bibinfo {year} {1945})}\BibitemShut {NoStop}%
\bibitem [{\citenamefont {Margolus}\ and\ \citenamefont {Levitin}(1998)}]{margolus1998maximum}%
  \BibitemOpen
  \bibfield  {author} {\bibinfo {author} {\bibfnamefont {N.}~\bibnamefont {Margolus}}\ and\ \bibinfo {author} {\bibfnamefont {L.~B.}\ \bibnamefont {Levitin}},\ }\href {\doibase https://doi.org/10.1016/S0167-2789(98)00054-2} {\bibfield  {journal} {\bibinfo  {journal} {Physica D}\ }\textbf {\bibinfo {volume} {120}},\ \bibinfo {pages} {188} (\bibinfo {year} {1998})}\BibitemShut {NoStop}%
\bibitem [{\citenamefont {del Campo}\ \emph {et~al.}(2013)\citenamefont {del Campo}, \citenamefont {Egusquiza}, \citenamefont {Plenio},\ and\ \citenamefont {Huelga}}]{PhysRevLett.110.050403}%
  \BibitemOpen
  \bibfield  {author} {\bibinfo {author} {\bibfnamefont {A.}~\bibnamefont {del Campo}}, \bibinfo {author} {\bibfnamefont {I.~L.}\ \bibnamefont {Egusquiza}}, \bibinfo {author} {\bibfnamefont {M.~B.}\ \bibnamefont {Plenio}}, \ and\ \bibinfo {author} {\bibfnamefont {S.~F.}\ \bibnamefont {Huelga}},\ }\href {\doibase 10.1103/PhysRevLett.110.050403} {\bibfield  {journal} {\bibinfo  {journal} {Phys. Rev. Lett.}\ }\textbf {\bibinfo {volume} {110}},\ \bibinfo {pages} {050403} (\bibinfo {year} {2013})}\BibitemShut {NoStop}%
\bibitem [{\citenamefont {Taddei}\ \emph {et~al.}(2013)\citenamefont {Taddei}, \citenamefont {Escher}, \citenamefont {Davidovich},\ and\ \citenamefont {de~Matos~Filho}}]{PhysRevLett.110.050402}%
  \BibitemOpen
  \bibfield  {author} {\bibinfo {author} {\bibfnamefont {M.~M.}\ \bibnamefont {Taddei}}, \bibinfo {author} {\bibfnamefont {B.~M.}\ \bibnamefont {Escher}}, \bibinfo {author} {\bibfnamefont {L.}~\bibnamefont {Davidovich}}, \ and\ \bibinfo {author} {\bibfnamefont {R.~L.}\ \bibnamefont {de~Matos~Filho}},\ }\href {\doibase 10.1103/PhysRevLett.110.050402} {\bibfield  {journal} {\bibinfo  {journal} {Phys. Rev. Lett.}\ }\textbf {\bibinfo {volume} {110}},\ \bibinfo {pages} {050402} (\bibinfo {year} {2013})}\BibitemShut {NoStop}%
\bibitem [{\citenamefont {Deffner}\ and\ \citenamefont {Lutz}(2013)}]{PhysRevLett.111.010402}%
  \BibitemOpen
  \bibfield  {author} {\bibinfo {author} {\bibfnamefont {S.}~\bibnamefont {Deffner}}\ and\ \bibinfo {author} {\bibfnamefont {E.}~\bibnamefont {Lutz}},\ }\href {\doibase 10.1103/PhysRevLett.111.010402} {\bibfield  {journal} {\bibinfo  {journal} {Phys. Rev. Lett.}\ }\textbf {\bibinfo {volume} {111}},\ \bibinfo {pages} {010402} (\bibinfo {year} {2013})}\BibitemShut {NoStop}%
\bibitem [{\citenamefont {Pires}\ \emph {et~al.}(2016)\citenamefont {Pires}, \citenamefont {Cianciaruso}, \citenamefont {C\'eleri}, \citenamefont {Adesso},\ and\ \citenamefont {Soares-Pinto}}]{pires2016generalized}%
  \BibitemOpen
  \bibfield  {author} {\bibinfo {author} {\bibfnamefont {D.~P.}\ \bibnamefont {Pires}}, \bibinfo {author} {\bibfnamefont {M.}~\bibnamefont {Cianciaruso}}, \bibinfo {author} {\bibfnamefont {L.~C.}\ \bibnamefont {C\'eleri}}, \bibinfo {author} {\bibfnamefont {G.}~\bibnamefont {Adesso}}, \ and\ \bibinfo {author} {\bibfnamefont {D.~O.}\ \bibnamefont {Soares-Pinto}},\ }\href {\doibase 10.1103/PhysRevX.6.021031} {\bibfield  {journal} {\bibinfo  {journal} {Phys. Rev. X}\ }\textbf {\bibinfo {volume} {6}},\ \bibinfo {pages} {021031} (\bibinfo {year} {2016})}\BibitemShut {NoStop}%
\bibitem [{\citenamefont {Campaioli}\ \emph {et~al.}(2018)\citenamefont {Campaioli}, \citenamefont {Pollock}, \citenamefont {Binder},\ and\ \citenamefont {Modi}}]{campaioli2018tightening}%
  \BibitemOpen
  \bibfield  {author} {\bibinfo {author} {\bibfnamefont {F.}~\bibnamefont {Campaioli}}, \bibinfo {author} {\bibfnamefont {F.~A.}\ \bibnamefont {Pollock}}, \bibinfo {author} {\bibfnamefont {F.~C.}\ \bibnamefont {Binder}}, \ and\ \bibinfo {author} {\bibfnamefont {K.}~\bibnamefont {Modi}},\ }\href {\doibase 10.1103/PhysRevLett.120.060409} {\bibfield  {journal} {\bibinfo  {journal} {Phys. Rev. Lett.}\ }\textbf {\bibinfo {volume} {120}},\ \bibinfo {pages} {060409} (\bibinfo {year} {2018})}\BibitemShut {NoStop}%
\bibitem [{\citenamefont {Campaioli}\ \emph {et~al.}(2019)\citenamefont {Campaioli}, \citenamefont {Pollock},\ and\ \citenamefont {Modi}}]{campaioli2019tight}%
  \BibitemOpen
  \bibfield  {author} {\bibinfo {author} {\bibfnamefont {F.}~\bibnamefont {Campaioli}}, \bibinfo {author} {\bibfnamefont {F.~A.}\ \bibnamefont {Pollock}}, \ and\ \bibinfo {author} {\bibfnamefont {K.}~\bibnamefont {Modi}},\ }\href {https://doi.org/10.22331/q-2019-08-05-168} {\bibfield  {journal} {\bibinfo  {journal} {Quantum}\ }\textbf {\bibinfo {volume} {3}},\ \bibinfo {pages} {168} (\bibinfo {year} {2019})}\BibitemShut {NoStop}%
\bibitem [{\citenamefont {O'Connor}\ \emph {et~al.}(2021)\citenamefont {O'Connor}, \citenamefont {Guarnieri},\ and\ \citenamefont {Campbell}}]{oconnor2020action}%
  \BibitemOpen
  \bibfield  {author} {\bibinfo {author} {\bibfnamefont {E.}~\bibnamefont {O'Connor}}, \bibinfo {author} {\bibfnamefont {G.}~\bibnamefont {Guarnieri}}, \ and\ \bibinfo {author} {\bibfnamefont {S.}~\bibnamefont {Campbell}},\ }\href {\doibase 10.1103/PhysRevA.103.022210} {\bibfield  {journal} {\bibinfo  {journal} {Phys. Rev. A}\ }\textbf {\bibinfo {volume} {103}},\ \bibinfo {pages} {022210} (\bibinfo {year} {2021})}\BibitemShut {NoStop}%
\bibitem [{\citenamefont {Van~Vu}\ and\ \citenamefont {Saito}(2023{\natexlab{a}})}]{vanvu2023topological}%
  \BibitemOpen
  \bibfield  {author} {\bibinfo {author} {\bibfnamefont {T.}~\bibnamefont {Van~Vu}}\ and\ \bibinfo {author} {\bibfnamefont {K.}~\bibnamefont {Saito}},\ }\href {\doibase 10.1103/PhysRevLett.130.010402} {\bibfield  {journal} {\bibinfo  {journal} {Phys. Rev. Lett.}\ }\textbf {\bibinfo {volume} {130}},\ \bibinfo {pages} {010402} (\bibinfo {year} {2023}{\natexlab{a}})}\BibitemShut {NoStop}%
\bibitem [{\citenamefont {Mai}\ and\ \citenamefont {Yu}(2023)}]{mai2023tight}%
  \BibitemOpen
  \bibfield  {author} {\bibinfo {author} {\bibfnamefont {Z.-y.}\ \bibnamefont {Mai}}\ and\ \bibinfo {author} {\bibfnamefont {C.-s.}\ \bibnamefont {Yu}},\ }\href {\doibase 10.1103/PhysRevA.108.052207} {\bibfield  {journal} {\bibinfo  {journal} {Phys. Rev. A}\ }\textbf {\bibinfo {volume} {108}},\ \bibinfo {pages} {052207} (\bibinfo {year} {2023})}\BibitemShut {NoStop}%
\bibitem [{\citenamefont {Lieb}\ and\ \citenamefont {Robinson}(1972)}]{lieb1972finite}%
  \BibitemOpen
  \bibfield  {author} {\bibinfo {author} {\bibfnamefont {E.~H.}\ \bibnamefont {Lieb}}\ and\ \bibinfo {author} {\bibfnamefont {D.~W.}\ \bibnamefont {Robinson}},\ }\href {https://doi.org/10.1007/BF01645779} {\bibfield  {journal} {\bibinfo  {journal} {Commun. Math. Phys.}\ }\textbf {\bibinfo {volume} {28}},\ \bibinfo {pages} {251} (\bibinfo {year} {1972})}\BibitemShut {NoStop}%
\bibitem [{\citenamefont {Garc\'{\i}a-Pintos}\ \emph {et~al.}(2022)\citenamefont {Garc\'{\i}a-Pintos}, \citenamefont {Nicholson}, \citenamefont {Green}, \citenamefont {del Campo},\ and\ \citenamefont {Gorshkov}}]{lgp2022unifying}%
  \BibitemOpen
  \bibfield  {author} {\bibinfo {author} {\bibfnamefont {L.~P.}\ \bibnamefont {Garc\'{\i}a-Pintos}}, \bibinfo {author} {\bibfnamefont {S.~B.}\ \bibnamefont {Nicholson}}, \bibinfo {author} {\bibfnamefont {J.~R.}\ \bibnamefont {Green}}, \bibinfo {author} {\bibfnamefont {A.}~\bibnamefont {del Campo}}, \ and\ \bibinfo {author} {\bibfnamefont {A.~V.}\ \bibnamefont {Gorshkov}},\ }\href {\doibase 10.1103/PhysRevX.12.011038} {\bibfield  {journal} {\bibinfo  {journal} {Phys. Rev. X}\ }\textbf {\bibinfo {volume} {12}},\ \bibinfo {pages} {011038} (\bibinfo {year} {2022})}\BibitemShut {NoStop}%
\bibitem [{\citenamefont {Carabba}\ \emph {et~al.}(2022)\citenamefont {Carabba}, \citenamefont {H{\"{o}}rnedal},\ and\ \citenamefont {del Campo}}]{Carabba2022quantumspeedlimits}%
  \BibitemOpen
  \bibfield  {author} {\bibinfo {author} {\bibfnamefont {N.}~\bibnamefont {Carabba}}, \bibinfo {author} {\bibfnamefont {N.}~\bibnamefont {H{\"{o}}rnedal}}, \ and\ \bibinfo {author} {\bibfnamefont {A.}~\bibnamefont {del Campo}},\ }\href {\doibase 10.22331/q-2022-12-22-884} {\bibfield  {journal} {\bibinfo  {journal} {{Quantum}}\ }\textbf {\bibinfo {volume} {6}},\ \bibinfo {pages} {884} (\bibinfo {year} {2022})}\BibitemShut {NoStop}%
\bibitem [{\citenamefont {Mohan}\ and\ \citenamefont {Pati}(2022)}]{mohan2022quantum}%
  \BibitemOpen
  \bibfield  {author} {\bibinfo {author} {\bibfnamefont {B.}~\bibnamefont {Mohan}}\ and\ \bibinfo {author} {\bibfnamefont {A.~K.}\ \bibnamefont {Pati}},\ }\href {\doibase 10.1103/PhysRevA.106.042436} {\bibfield  {journal} {\bibinfo  {journal} {Phys. Rev. A}\ }\textbf {\bibinfo {volume} {106}},\ \bibinfo {pages} {042436} (\bibinfo {year} {2022})}\BibitemShut {NoStop}%
\bibitem [{\citenamefont {Hörnedal}\ \emph {et~al.}(2022)\citenamefont {Hörnedal}, \citenamefont {Carabba}, \citenamefont {Matsoukas-Roubeas},\ and\ \citenamefont {del Campo}}]{H_rnedal_2022}%
  \BibitemOpen
  \bibfield  {author} {\bibinfo {author} {\bibfnamefont {N.}~\bibnamefont {Hörnedal}}, \bibinfo {author} {\bibfnamefont {N.}~\bibnamefont {Carabba}}, \bibinfo {author} {\bibfnamefont {A.~S.}\ \bibnamefont {Matsoukas-Roubeas}}, \ and\ \bibinfo {author} {\bibfnamefont {A.}~\bibnamefont {del Campo}},\ }\href {\doibase 10.1038/s42005-022-00985-1} {\bibfield  {journal} {\bibinfo  {journal} {Commun. Phys.}\ }\textbf {\bibinfo {volume} {5}} (\bibinfo {year} {2022}),\ 10.1038/s42005-022-00985-1}\BibitemShut {NoStop}%
\bibitem [{\citenamefont {Maleki}\ \emph {et~al.}(2023)\citenamefont {Maleki}, \citenamefont {Ahansaz},\ and\ \citenamefont {Maleki}}]{maleki2023speed}%
  \BibitemOpen
  \bibfield  {author} {\bibinfo {author} {\bibfnamefont {Y.}~\bibnamefont {Maleki}}, \bibinfo {author} {\bibfnamefont {B.}~\bibnamefont {Ahansaz}}, \ and\ \bibinfo {author} {\bibfnamefont {A.}~\bibnamefont {Maleki}},\ }\href {https://doi.org/10.1038/s41598-023-39082-w} {\bibfield  {journal} {\bibinfo  {journal} {Sci. Rep.}\ }\textbf {\bibinfo {volume} {13}},\ \bibinfo {pages} {12031} (\bibinfo {year} {2023})}\BibitemShut {NoStop}%
\bibitem [{\citenamefont {Herb}\ and\ \citenamefont {Degen}(2024)}]{herb2024quantum}%
  \BibitemOpen
  \bibfield  {author} {\bibinfo {author} {\bibfnamefont {K.}~\bibnamefont {Herb}}\ and\ \bibinfo {author} {\bibfnamefont {C.~L.}\ \bibnamefont {Degen}},\ }\href {https://doi.org/10.48550/arXiv.2406.18348} {\bibfield  {journal} {\bibinfo  {journal} {arXiv preprint arXiv:2406.18348}\ } (\bibinfo {year} {2024})}\BibitemShut {NoStop}%
\bibitem [{\citenamefont {Campbell}\ \emph {et~al.}(2018)\citenamefont {Campbell}, \citenamefont {Genoni},\ and\ \citenamefont {Deffner}}]{campbell2018precision}%
  \BibitemOpen
  \bibfield  {author} {\bibinfo {author} {\bibfnamefont {S.}~\bibnamefont {Campbell}}, \bibinfo {author} {\bibfnamefont {M.~G.}\ \bibnamefont {Genoni}}, \ and\ \bibinfo {author} {\bibfnamefont {S.}~\bibnamefont {Deffner}},\ }\href {https://doi.org/10.1088/2058-9565/aaa641} {\bibfield  {journal} {\bibinfo  {journal} {Quantum Sci. Technol.}\ }\textbf {\bibinfo {volume} {3}},\ \bibinfo {pages} {025002} (\bibinfo {year} {2018})}\BibitemShut {NoStop}%
\bibitem [{\citenamefont {Campaioli}\ \emph {et~al.}(2017)\citenamefont {Campaioli}, \citenamefont {Pollock}, \citenamefont {Binder}, \citenamefont {C\'eleri}, \citenamefont {Goold}, \citenamefont {Vinjanampathy},\ and\ \citenamefont {Modi}}]{campaioli2017enhancing}%
  \BibitemOpen
  \bibfield  {author} {\bibinfo {author} {\bibfnamefont {F.}~\bibnamefont {Campaioli}}, \bibinfo {author} {\bibfnamefont {F.~A.}\ \bibnamefont {Pollock}}, \bibinfo {author} {\bibfnamefont {F.~C.}\ \bibnamefont {Binder}}, \bibinfo {author} {\bibfnamefont {L.}~\bibnamefont {C\'eleri}}, \bibinfo {author} {\bibfnamefont {J.}~\bibnamefont {Goold}}, \bibinfo {author} {\bibfnamefont {S.}~\bibnamefont {Vinjanampathy}}, \ and\ \bibinfo {author} {\bibfnamefont {K.}~\bibnamefont {Modi}},\ }\href {\doibase 10.1103/PhysRevLett.118.150601} {\bibfield  {journal} {\bibinfo  {journal} {Phys. Rev. Lett.}\ }\textbf {\bibinfo {volume} {118}},\ \bibinfo {pages} {150601} (\bibinfo {year} {2017})}\BibitemShut {NoStop}%
\bibitem [{\citenamefont {Garc\'{\i}a-Pintos}\ \emph {et~al.}(2020)\citenamefont {Garc\'{\i}a-Pintos}, \citenamefont {Hamma},\ and\ \citenamefont {del Campo}}]{PhysRevLett.125.040601}%
  \BibitemOpen
  \bibfield  {author} {\bibinfo {author} {\bibfnamefont {L.~P.}\ \bibnamefont {Garc\'{\i}a-Pintos}}, \bibinfo {author} {\bibfnamefont {A.}~\bibnamefont {Hamma}}, \ and\ \bibinfo {author} {\bibfnamefont {A.}~\bibnamefont {del Campo}},\ }\href {\doibase 10.1103/PhysRevLett.125.040601} {\bibfield  {journal} {\bibinfo  {journal} {Phys. Rev. Lett.}\ }\textbf {\bibinfo {volume} {125}},\ \bibinfo {pages} {040601} (\bibinfo {year} {2020})}\BibitemShut {NoStop}%
\bibitem [{\citenamefont {Funo}\ \emph {et~al.}(2019)\citenamefont {Funo}, \citenamefont {Shiraishi},\ and\ \citenamefont {Saito}}]{funo2019speed}%
  \BibitemOpen
  \bibfield  {author} {\bibinfo {author} {\bibfnamefont {K.}~\bibnamefont {Funo}}, \bibinfo {author} {\bibfnamefont {N.}~\bibnamefont {Shiraishi}}, \ and\ \bibinfo {author} {\bibfnamefont {K.}~\bibnamefont {Saito}},\ }\href {\doibase 10.1088/1367-2630/aaf9f5} {\bibfield  {journal} {\bibinfo  {journal} {New J. Phys.}\ }\textbf {\bibinfo {volume} {21}},\ \bibinfo {pages} {013006} (\bibinfo {year} {2019})}\BibitemShut {NoStop}%
\bibitem [{\citenamefont {Das}\ \emph {et~al.}(2021)\citenamefont {Das}, \citenamefont {Bera}, \citenamefont {Chakraborty},\ and\ \citenamefont {Chru\ifmmode \acute{s}\else \'{s}\fi{}ci\ifmmode~\acute{n}\else \'{n}\fi{}ski}}]{das2021thermo}%
  \BibitemOpen
  \bibfield  {author} {\bibinfo {author} {\bibfnamefont {A.}~\bibnamefont {Das}}, \bibinfo {author} {\bibfnamefont {A.}~\bibnamefont {Bera}}, \bibinfo {author} {\bibfnamefont {S.}~\bibnamefont {Chakraborty}}, \ and\ \bibinfo {author} {\bibfnamefont {D.}~\bibnamefont {Chru\ifmmode \acute{s}\else \'{s}\fi{}ci\ifmmode~\acute{n}\else \'{n}\fi{}ski}},\ }\href {\doibase 10.1103/PhysRevA.104.042202} {\bibfield  {journal} {\bibinfo  {journal} {Phys. Rev. A}\ }\textbf {\bibinfo {volume} {104}},\ \bibinfo {pages} {042202} (\bibinfo {year} {2021})}\BibitemShut {NoStop}%
\bibitem [{\citenamefont {Van~Vu}\ and\ \citenamefont {Saito}(2023{\natexlab{b}})}]{vanvu2023thermo}%
  \BibitemOpen
  \bibfield  {author} {\bibinfo {author} {\bibfnamefont {T.}~\bibnamefont {Van~Vu}}\ and\ \bibinfo {author} {\bibfnamefont {K.}~\bibnamefont {Saito}},\ }\href {\doibase 10.1103/PhysRevX.13.011013} {\bibfield  {journal} {\bibinfo  {journal} {Phys. Rev. X}\ }\textbf {\bibinfo {volume} {13}},\ \bibinfo {pages} {011013} (\bibinfo {year} {2023}{\natexlab{b}})}\BibitemShut {NoStop}%
\bibitem [{\citenamefont {Caneva}\ \emph {et~al.}(2009)\citenamefont {Caneva}, \citenamefont {Murphy}, \citenamefont {Calarco}, \citenamefont {Fazio}, \citenamefont {Montangero}, \citenamefont {Giovannetti},\ and\ \citenamefont {Santoro}}]{caneva2009optimal}%
  \BibitemOpen
  \bibfield  {author} {\bibinfo {author} {\bibfnamefont {T.}~\bibnamefont {Caneva}}, \bibinfo {author} {\bibfnamefont {M.}~\bibnamefont {Murphy}}, \bibinfo {author} {\bibfnamefont {T.}~\bibnamefont {Calarco}}, \bibinfo {author} {\bibfnamefont {R.}~\bibnamefont {Fazio}}, \bibinfo {author} {\bibfnamefont {S.}~\bibnamefont {Montangero}}, \bibinfo {author} {\bibfnamefont {V.}~\bibnamefont {Giovannetti}}, \ and\ \bibinfo {author} {\bibfnamefont {G.~E.}\ \bibnamefont {Santoro}},\ }\href {\doibase 10.1103/PhysRevLett.103.240501} {\bibfield  {journal} {\bibinfo  {journal} {Phys. Rev. Lett.}\ }\textbf {\bibinfo {volume} {103}},\ \bibinfo {pages} {240501} (\bibinfo {year} {2009})}\BibitemShut {NoStop}%
\bibitem [{\citenamefont {Garc\'{\i}a-Pintos}\ \emph {et~al.}(2023)\citenamefont {Garc\'{\i}a-Pintos}, \citenamefont {Brady}, \citenamefont {Bringewatt},\ and\ \citenamefont {Liu}}]{lgp2023lower}%
  \BibitemOpen
  \bibfield  {author} {\bibinfo {author} {\bibfnamefont {L.~P.}\ \bibnamefont {Garc\'{\i}a-Pintos}}, \bibinfo {author} {\bibfnamefont {L.~T.}\ \bibnamefont {Brady}}, \bibinfo {author} {\bibfnamefont {J.}~\bibnamefont {Bringewatt}}, \ and\ \bibinfo {author} {\bibfnamefont {Y.-K.}\ \bibnamefont {Liu}},\ }\href {\doibase 10.1103/PhysRevLett.130.140601} {\bibfield  {journal} {\bibinfo  {journal} {Phys. Rev. Lett.}\ }\textbf {\bibinfo {volume} {130}},\ \bibinfo {pages} {140601} (\bibinfo {year} {2023})}\BibitemShut {NoStop}%
\bibitem [{\citenamefont {Deffner}\ and\ \citenamefont {Campbell}(2017)}]{deffner2017quantum}%
  \BibitemOpen
  \bibfield  {author} {\bibinfo {author} {\bibfnamefont {S.}~\bibnamefont {Deffner}}\ and\ \bibinfo {author} {\bibfnamefont {S.}~\bibnamefont {Campbell}},\ }\href {https://doi.org/10.1088/1751-8121/aa86c6} {\bibfield  {journal} {\bibinfo  {journal} {J. Phys. A: Math. Theor.}\ }\textbf {\bibinfo {volume} {50}},\ \bibinfo {pages} {453001} (\bibinfo {year} {2017})}\BibitemShut {NoStop}%
\bibitem [{\citenamefont {Petz}(1996)}]{petz1996monotone}%
  \BibitemOpen
  \bibfield  {author} {\bibinfo {author} {\bibfnamefont {D.}~\bibnamefont {Petz}},\ }\href {https://doi.org/10.1016/0024-3795(94)00211-8} {\bibfield  {journal} {\bibinfo  {journal} {Linear Algebra Its Appl.}\ }\textbf {\bibinfo {volume} {244}},\ \bibinfo {pages} {81} (\bibinfo {year} {1996})}\BibitemShut {NoStop}%
\bibitem [{\citenamefont {Petz}(2002)}]{petz2002covariance}%
  \BibitemOpen
  \bibfield  {author} {\bibinfo {author} {\bibfnamefont {D.}~\bibnamefont {Petz}},\ }\href {https://doi.org/10.1088/0305-4470/35/4/305} {\bibfield  {journal} {\bibinfo  {journal} {J. Phys. A: Math. Gen.}\ }\textbf {\bibinfo {volume} {35}},\ \bibinfo {pages} {929} (\bibinfo {year} {2002})}\BibitemShut {NoStop}%
\bibitem [{Note1()}]{Note1}%
  \BibitemOpen
  \bibinfo {note} {In this paper, we only consider positive-definite states. A subset of the metrics considered here admit an extension to pure states, but, here, they all correspond to the Fubini-Study metric~\cite {petz1996geometries}, so generalized geometric speed limits for pure states are uninteresting. There is a large family of non-positive definite, non-pure states, but, to our knowledge, metrics on this full topological boundary have not been considered in the literature.}\BibitemShut {Stop}%
\bibitem [{Note2()}]{Note2}%
  \BibitemOpen
  \bibinfo {note} {While monotonicity is a natural requirement for a metric, one can define geometric speed limits without this requirement~\cite {campaioli2019tight}.}\BibitemShut {Stop}%
\bibitem [{Note3()}]{Note3}%
  \BibitemOpen
  \bibinfo {note} {The definition can be adapted slightly to allow for non-commuting operators, but as we do not need this generalization we keep the mean in its simplest form.}\BibitemShut {Stop}%
\bibitem [{\citenamefont {Kubo}\ and\ \citenamefont {Ando}(1980)}]{kubo1980means}%
  \BibitemOpen
  \bibfield  {author} {\bibinfo {author} {\bibfnamefont {F.}~\bibnamefont {Kubo}}\ and\ \bibinfo {author} {\bibfnamefont {T.}~\bibnamefont {Ando}},\ }\href {https://doi.org/10.1007/BF01371042} {\bibfield  {journal} {\bibinfo  {journal} {Math. Ann.}\ }\textbf {\bibinfo {volume} {246}},\ \bibinfo {pages} {205} (\bibinfo {year} {1980})}\BibitemShut {NoStop}%
\bibitem [{SM()}]{SM}%
  \BibitemOpen
  \href@noop {} {}\bibinfo {note} {See supplemental material.}\BibitemShut {Stop}%
\bibitem [{Note4()}]{Note4}%
  \BibitemOpen
  \bibinfo {note} {The normalization factor of $1/4$ in the definition of $g_\rho $ is the conventional choice as it makes $g_\rho ^f$ coincide with the Euclidean metric restricted to the positive orthant of the unit sphere defined in the space of eigenvalues of $\rho $.}\BibitemShut {Stop}%
\bibitem [{\citenamefont {Petz}\ and\ \citenamefont {Sud{\'a}r}(1996)}]{petz1996geometries}%
  \BibitemOpen
  \bibfield  {author} {\bibinfo {author} {\bibfnamefont {D.}~\bibnamefont {Petz}}\ and\ \bibinfo {author} {\bibfnamefont {C.}~\bibnamefont {Sud{\'a}r}},\ }\href {https://doi.org/10.1063/1.531535} {\bibfield  {journal} {\bibinfo  {journal} {J. Math. Phys.}\ }\textbf {\bibinfo {volume} {37}},\ \bibinfo {pages} {2662} (\bibinfo {year} {1996})}\BibitemShut {NoStop}%
\bibitem [{\citenamefont {Gibilisco}(2024)}]{gibilisco2024f}%
  \BibitemOpen
  \bibfield  {author} {\bibinfo {author} {\bibfnamefont {P.}~\bibnamefont {Gibilisco}},\ }\href {https://doi.org/10.3390/e26040286} {\bibfield  {journal} {\bibinfo  {journal} {Entropy}\ }\textbf {\bibinfo {volume} {26}},\ \bibinfo {pages} {286} (\bibinfo {year} {2024})}\BibitemShut {NoStop}%
\bibitem [{\citenamefont {Gibilisco}\ and\ \citenamefont {Isola}(2004)}]{gibilisco2004characterisation}%
  \BibitemOpen
  \bibfield  {author} {\bibinfo {author} {\bibfnamefont {P.}~\bibnamefont {Gibilisco}}\ and\ \bibinfo {author} {\bibfnamefont {T.}~\bibnamefont {Isola}},\ }\href {https://doi.org/10.1007/BF02530551} {\bibfield  {journal} {\bibinfo  {journal} {Ann. Inst. Stat. Math.}\ }\textbf {\bibinfo {volume} {56}},\ \bibinfo {pages} {369} (\bibinfo {year} {2004})}\BibitemShut {NoStop}%
\bibitem [{cod()}]{code}%
  \BibitemOpen
  \href@noop {} {\enquote {\bibinfo {title} {Code available upon request.}}\ }\BibitemShut {NoStop}%
\bibitem [{Note5()}]{Note5}%
  \BibitemOpen
  \bibinfo {note} {Note that this gives us flexibility to fix the Hamiltonian norm.}\BibitemShut {Stop}%
\bibitem [{\citenamefont {Magann}\ \emph {et~al.}(2022)\citenamefont {Magann}, \citenamefont {Rudinger}, \citenamefont {Grace},\ and\ \citenamefont {Sarovar}}]{magann2022feedback}%
  \BibitemOpen
  \bibfield  {author} {\bibinfo {author} {\bibfnamefont {A.~B.}\ \bibnamefont {Magann}}, \bibinfo {author} {\bibfnamefont {K.~M.}\ \bibnamefont {Rudinger}}, \bibinfo {author} {\bibfnamefont {M.~D.}\ \bibnamefont {Grace}}, \ and\ \bibinfo {author} {\bibfnamefont {M.}~\bibnamefont {Sarovar}},\ }\href {https://doi.org/10.1103/PhysRevLett.129.250502} {\bibfield  {journal} {\bibinfo  {journal} {Phys. Rev. Lett.}\ }\textbf {\bibinfo {volume} {129}},\ \bibinfo {pages} {250502} (\bibinfo {year} {2022})}\BibitemShut {NoStop}%
\bibitem [{\citenamefont {Stokes}\ \emph {et~al.}(2020)\citenamefont {Stokes}, \citenamefont {Izaac}, \citenamefont {Killoran},\ and\ \citenamefont {Carleo}}]{stokes2020quantum}%
  \BibitemOpen
  \bibfield  {author} {\bibinfo {author} {\bibfnamefont {J.}~\bibnamefont {Stokes}}, \bibinfo {author} {\bibfnamefont {J.}~\bibnamefont {Izaac}}, \bibinfo {author} {\bibfnamefont {N.}~\bibnamefont {Killoran}}, \ and\ \bibinfo {author} {\bibfnamefont {G.}~\bibnamefont {Carleo}},\ }\href {https://doi.org/10.22331/q-2020-05-25-269} {\bibfield  {journal} {\bibinfo  {journal} {Quantum}\ }\textbf {\bibinfo {volume} {4}},\ \bibinfo {pages} {269} (\bibinfo {year} {2020})}\BibitemShut {NoStop}%
\bibitem [{\citenamefont {Lu}\ \emph {et~al.}(2020)\citenamefont {Lu}, \citenamefont {Ma},\ and\ \citenamefont {Zhang}}]{lu2020generalized}%
  \BibitemOpen
  \bibfield  {author} {\bibinfo {author} {\bibfnamefont {X.-M.}\ \bibnamefont {Lu}}, \bibinfo {author} {\bibfnamefont {Z.}~\bibnamefont {Ma}}, \ and\ \bibinfo {author} {\bibfnamefont {C.}~\bibnamefont {Zhang}},\ }\href {\doibase 10.1103/PhysRevA.101.022303} {\bibfield  {journal} {\bibinfo  {journal} {Phys. Rev. A}\ }\textbf {\bibinfo {volume} {101}},\ \bibinfo {pages} {022303} (\bibinfo {year} {2020})}\BibitemShut {NoStop}%
\bibitem [{\citenamefont {Gibilisco}\ \emph {et~al.}(2008)\citenamefont {Gibilisco}, \citenamefont {Imparato},\ and\ \citenamefont {Isola}}]{gibilisco2008robertson}%
  \BibitemOpen
  \bibfield  {author} {\bibinfo {author} {\bibfnamefont {P.}~\bibnamefont {Gibilisco}}, \bibinfo {author} {\bibfnamefont {D.}~\bibnamefont {Imparato}}, \ and\ \bibinfo {author} {\bibfnamefont {T.}~\bibnamefont {Isola}},\ }\href {https://doi.org/10.1016/j.laa.2007.10.013} {\bibfield  {journal} {\bibinfo  {journal} {Linear Algebra Its Appl.}\ }\textbf {\bibinfo {volume} {428}},\ \bibinfo {pages} {1706} (\bibinfo {year} {2008})}\BibitemShut {NoStop}%
\bibitem [{\citenamefont {Yanagi}(2011)}]{yanagi2011metric}%
  \BibitemOpen
  \bibfield  {author} {\bibinfo {author} {\bibfnamefont {K.}~\bibnamefont {Yanagi}},\ }\href {https://doi.org/10.1016/j.jmaa.2011.03.068} {\bibfield  {journal} {\bibinfo  {journal} {J. Math. Anal. Appl.}\ }\textbf {\bibinfo {volume} {380}},\ \bibinfo {pages} {888} (\bibinfo {year} {2011})}\BibitemShut {NoStop}%
\bibitem [{\citenamefont {Pires}\ \emph {et~al.}(2024)\citenamefont {Pires}, \citenamefont {deAzevedo}, \citenamefont {Soares-Pinto}, \citenamefont {Brito},\ and\ \citenamefont {Filgueiras}}]{pires2024experimental}%
  \BibitemOpen
  \bibfield  {author} {\bibinfo {author} {\bibfnamefont {D.~P.}\ \bibnamefont {Pires}}, \bibinfo {author} {\bibfnamefont {E.~R.}\ \bibnamefont {deAzevedo}}, \bibinfo {author} {\bibfnamefont {D.~O.}\ \bibnamefont {Soares-Pinto}}, \bibinfo {author} {\bibfnamefont {F.}~\bibnamefont {Brito}}, \ and\ \bibinfo {author} {\bibfnamefont {J.~G.}\ \bibnamefont {Filgueiras}},\ }\href {https://doi.org/10.1038/s42005-024-01634-5} {\bibfield  {journal} {\bibinfo  {journal} {Commun. Phys.}\ }\textbf {\bibinfo {volume} {7}},\ \bibinfo {pages} {142} (\bibinfo {year} {2024})}\BibitemShut {NoStop}%
\bibitem [{\citenamefont {Meyer}(2021)}]{meyer2021fisher}%
  \BibitemOpen
  \bibfield  {author} {\bibinfo {author} {\bibfnamefont {J.~J.}\ \bibnamefont {Meyer}},\ }\href {https://doi.org/10.22331/q-2021-09-09-539} {\bibfield  {journal} {\bibinfo  {journal} {Quantum}\ }\textbf {\bibinfo {volume} {5}},\ \bibinfo {pages} {539} (\bibinfo {year} {2021})}\BibitemShut {NoStop}%
\bibitem [{\citenamefont {Abbas}\ \emph {et~al.}(2021)\citenamefont {Abbas}, \citenamefont {Sutter}, \citenamefont {Zoufal}, \citenamefont {Lucchi}, \citenamefont {Figalli},\ and\ \citenamefont {Woerner}}]{abbas2021power}%
  \BibitemOpen
  \bibfield  {author} {\bibinfo {author} {\bibfnamefont {A.}~\bibnamefont {Abbas}}, \bibinfo {author} {\bibfnamefont {D.}~\bibnamefont {Sutter}}, \bibinfo {author} {\bibfnamefont {C.}~\bibnamefont {Zoufal}}, \bibinfo {author} {\bibfnamefont {A.}~\bibnamefont {Lucchi}}, \bibinfo {author} {\bibfnamefont {A.}~\bibnamefont {Figalli}}, \ and\ \bibinfo {author} {\bibfnamefont {S.}~\bibnamefont {Woerner}},\ }\href@noop {} {\bibfield  {journal} {\bibinfo  {journal} {Nat. Comput. Sci.}\ }\textbf {\bibinfo {volume} {1}},\ \bibinfo {pages} {403} (\bibinfo {year} {2021})}\BibitemShut {NoStop}%
\bibitem [{\citenamefont {Haug}\ \emph {et~al.}(2021)\citenamefont {Haug}, \citenamefont {Bharti},\ and\ \citenamefont {Kim}}]{haug2021capacity}%
  \BibitemOpen
  \bibfield  {author} {\bibinfo {author} {\bibfnamefont {T.}~\bibnamefont {Haug}}, \bibinfo {author} {\bibfnamefont {K.}~\bibnamefont {Bharti}}, \ and\ \bibinfo {author} {\bibfnamefont {M.}~\bibnamefont {Kim}},\ }\href {https://doi.org/10.1103/PRXQuantum.2.040309} {\bibfield  {journal} {\bibinfo  {journal} {PRX Quantum}\ }\textbf {\bibinfo {volume} {2}},\ \bibinfo {pages} {040309} (\bibinfo {year} {2021})}\BibitemShut {NoStop}%
\bibitem [{\citenamefont {Garc\'{\i}a-Pintos}\ \emph {et~al.}(2024)\citenamefont {Garc\'{\i}a-Pintos}, \citenamefont {Bharti}, \citenamefont {Bringewatt}, \citenamefont {Dehghani}, \citenamefont {Ehrenberg}, \citenamefont {Yunger~Halpern},\ and\ \citenamefont {Gorshkov}}]{garcia2024estimation}%
  \BibitemOpen
  \bibfield  {author} {\bibinfo {author} {\bibfnamefont {L.~P.}\ \bibnamefont {Garc\'{\i}a-Pintos}}, \bibinfo {author} {\bibfnamefont {K.}~\bibnamefont {Bharti}}, \bibinfo {author} {\bibfnamefont {J.}~\bibnamefont {Bringewatt}}, \bibinfo {author} {\bibfnamefont {H.}~\bibnamefont {Dehghani}}, \bibinfo {author} {\bibfnamefont {A.}~\bibnamefont {Ehrenberg}}, \bibinfo {author} {\bibfnamefont {N.}~\bibnamefont {Yunger~Halpern}}, \ and\ \bibinfo {author} {\bibfnamefont {A.~V.}\ \bibnamefont {Gorshkov}},\ }\href {\doibase 10.1103/PhysRevLett.133.040802} {\bibfield  {journal} {\bibinfo  {journal} {Phys. Rev. Lett.}\ }\textbf {\bibinfo {volume} {133}},\ \bibinfo {pages} {040802} (\bibinfo {year} {2024})}\BibitemShut {NoStop}%
\bibitem [{Note6()}]{Note6}%
  \BibitemOpen
  \bibinfo {note} {Just use that $m^f(p_i,p_i)=p_i$ for all $f$}\BibitemShut {NoStop}%
\bibitem [{\citenamefont {Bianchetti}\ \emph {et~al.}(2010)\citenamefont {Bianchetti}, \citenamefont {Filipp}, \citenamefont {Baur}, \citenamefont {Fink}, \citenamefont {Lang}, \citenamefont {Steffen}, \citenamefont {Boissonneault}, \citenamefont {Blais},\ and\ \citenamefont {Wallraff}}]{Wallraff_qutrit}%
  \BibitemOpen
  \bibfield  {author} {\bibinfo {author} {\bibfnamefont {R.}~\bibnamefont {Bianchetti}}, \bibinfo {author} {\bibfnamefont {S.}~\bibnamefont {Filipp}}, \bibinfo {author} {\bibfnamefont {M.}~\bibnamefont {Baur}}, \bibinfo {author} {\bibfnamefont {J.~M.}\ \bibnamefont {Fink}}, \bibinfo {author} {\bibfnamefont {C.}~\bibnamefont {Lang}}, \bibinfo {author} {\bibfnamefont {L.}~\bibnamefont {Steffen}}, \bibinfo {author} {\bibfnamefont {M.}~\bibnamefont {Boissonneault}}, \bibinfo {author} {\bibfnamefont {A.}~\bibnamefont {Blais}}, \ and\ \bibinfo {author} {\bibfnamefont {A.}~\bibnamefont {Wallraff}},\ }\href {\doibase 10.1103/PhysRevLett.105.223601} {\bibfield  {journal} {\bibinfo  {journal} {Phys. Rev. Lett.}\ }\textbf {\bibinfo {volume} {105}},\ \bibinfo {pages} {223601} (\bibinfo {year} {2010})}\BibitemShut {NoStop}%
\bibitem [{\citenamefont {Blais}\ \emph {et~al.}(2021)\citenamefont {Blais}, \citenamefont {Grimsmo}, \citenamefont {Girvin},\ and\ \citenamefont {Wallraff}}]{Blais2021_cQED}%
  \BibitemOpen
  \bibfield  {author} {\bibinfo {author} {\bibfnamefont {A.}~\bibnamefont {Blais}}, \bibinfo {author} {\bibfnamefont {A.~L.}\ \bibnamefont {Grimsmo}}, \bibinfo {author} {\bibfnamefont {S.}~\bibnamefont {Girvin}}, \ and\ \bibinfo {author} {\bibfnamefont {A.}~\bibnamefont {Wallraff}},\ }\href {\doibase 10.1103/RevModPhys.93.025005} {\bibfield  {journal} {\bibinfo  {journal} {Rev. Mod. Phys.}\ }\textbf {\bibinfo {volume} {93}},\ \bibinfo {pages} {025005} (\bibinfo {year} {2021})}\BibitemShut {NoStop}%
\bibitem [{\citenamefont {Krantz}\ \emph {et~al.}(2019)\citenamefont {Krantz}, \citenamefont {Kjaergaard}, \citenamefont {Yan}, \citenamefont {Orlando}, \citenamefont {Gustavsson},\ and\ \citenamefont {Oliver}}]{krantz2019A}%
  \BibitemOpen
  \bibfield  {author} {\bibinfo {author} {\bibfnamefont {P.}~\bibnamefont {Krantz}}, \bibinfo {author} {\bibfnamefont {M.}~\bibnamefont {Kjaergaard}}, \bibinfo {author} {\bibfnamefont {F.}~\bibnamefont {Yan}}, \bibinfo {author} {\bibfnamefont {T.~P.}\ \bibnamefont {Orlando}}, \bibinfo {author} {\bibfnamefont {S.}~\bibnamefont {Gustavsson}}, \ and\ \bibinfo {author} {\bibfnamefont {W.~D.}\ \bibnamefont {Oliver}},\ }\href {\doibase 10.1063/1.5089550} {\bibfield  {journal} {\bibinfo  {journal} {Applied Physics Reviews}\ }\textbf {\bibinfo {volume} {6}},\ \bibinfo {pages} {021318} (\bibinfo {year} {2019})}\BibitemShut {NoStop}%
\bibitem [{\citenamefont {Huang}\ \emph {et~al.}(2023)\citenamefont {Huang}, \citenamefont {Huang}, \citenamefont {Wang}, \citenamefont {Steffen}, \citenamefont {Cripe}, \citenamefont {Wellstood},\ and\ \citenamefont {Palmer}}]{Huang_packaging}%
  \BibitemOpen
  \bibfield  {author} {\bibinfo {author} {\bibfnamefont {Y.}~\bibnamefont {Huang}}, \bibinfo {author} {\bibfnamefont {Y.-H.}\ \bibnamefont {Huang}}, \bibinfo {author} {\bibfnamefont {H.}~\bibnamefont {Wang}}, \bibinfo {author} {\bibfnamefont {Z.}~\bibnamefont {Steffen}}, \bibinfo {author} {\bibfnamefont {J.}~\bibnamefont {Cripe}}, \bibinfo {author} {\bibfnamefont {F.~C.}\ \bibnamefont {Wellstood}}, \ and\ \bibinfo {author} {\bibfnamefont {B.~S.}\ \bibnamefont {Palmer}},\ }\href {\doibase 10.1063/5.0155053} {\bibfield  {journal} {\bibinfo  {journal} {Appl. Phys. Lett}\ }\textbf {\bibinfo {volume} {123}},\ \bibinfo {pages} {044001} (\bibinfo {year} {2023})}\BibitemShut {NoStop}%
\bibitem [{\citenamefont {Bateman}(1910)}]{bateman_eq}%
  \BibitemOpen
  \bibfield  {author} {\bibinfo {author} {\bibfnamefont {H.}~\bibnamefont {Bateman}},\ }\href {https://archive.org/details/cbarchive_122715_solutionofasystemofdifferentia1843} {\bibfield  {journal} {\bibinfo  {journal} {Proceedings of the Cambridge Philosophical Society, Mathematical and physical sciences}\ }\textbf {\bibinfo {volume} {15}} (\bibinfo {year} {1908-1910})}\BibitemShut {NoStop}%
\bibitem [{\citenamefont {Peterer}\ \emph {et~al.}(2015)\citenamefont {Peterer}, \citenamefont {Bader}, \citenamefont {Jin}, \citenamefont {Yan}, \citenamefont {Kamal}, \citenamefont {Gudmundsen}, \citenamefont {Leek}, \citenamefont {Orlando}, \citenamefont {Oliver},\ and\ \citenamefont {Gustavsson}}]{f_state_decay}%
  \BibitemOpen
  \bibfield  {author} {\bibinfo {author} {\bibfnamefont {M.~J.}\ \bibnamefont {Peterer}}, \bibinfo {author} {\bibfnamefont {S.~J.}\ \bibnamefont {Bader}}, \bibinfo {author} {\bibfnamefont {X.}~\bibnamefont {Jin}}, \bibinfo {author} {\bibfnamefont {F.}~\bibnamefont {Yan}}, \bibinfo {author} {\bibfnamefont {A.}~\bibnamefont {Kamal}}, \bibinfo {author} {\bibfnamefont {T.~J.}\ \bibnamefont {Gudmundsen}}, \bibinfo {author} {\bibfnamefont {P.~J.}\ \bibnamefont {Leek}}, \bibinfo {author} {\bibfnamefont {T.~P.}\ \bibnamefont {Orlando}}, \bibinfo {author} {\bibfnamefont {W.~D.}\ \bibnamefont {Oliver}}, \ and\ \bibinfo {author} {\bibfnamefont {S.}~\bibnamefont {Gustavsson}},\ }\href {\doibase 10.1103/PhysRevLett.114.010501} {\bibfield  {journal} {\bibinfo  {journal} {Phys. Rev. Lett.}\ }\textbf {\bibinfo {volume} {114}},\ \bibinfo {pages} {010501} (\bibinfo {year} {2015})}\BibitemShut {NoStop}%
\bibitem [{\citenamefont {Liu}\ \emph {et~al.}(2019)\citenamefont {Liu}, \citenamefont {Yuan}, \citenamefont {Lu},\ and\ \citenamefont {Wang}}]{Liu_2019}%
  \BibitemOpen
  \bibfield  {author} {\bibinfo {author} {\bibfnamefont {J.}~\bibnamefont {Liu}}, \bibinfo {author} {\bibfnamefont {H.}~\bibnamefont {Yuan}}, \bibinfo {author} {\bibfnamefont {X.-M.}\ \bibnamefont {Lu}}, \ and\ \bibinfo {author} {\bibfnamefont {X.}~\bibnamefont {Wang}},\ }\href {\doibase 10.1088/1751-8121/ab5d4d} {\bibfield  {journal} {\bibinfo  {journal} {J. Phys. A: Math. and Theor.}\ }\textbf {\bibinfo {volume} {53}},\ \bibinfo {pages} {023001} (\bibinfo {year} {2019})}\BibitemShut {NoStop}%
\bibitem [{\citenamefont {Boixo}\ \emph {et~al.}(2007)\citenamefont {Boixo}, \citenamefont {Flammia}, \citenamefont {Caves},\ and\ \citenamefont {Geremia}}]{boixo2008generalized}%
  \BibitemOpen
  \bibfield  {author} {\bibinfo {author} {\bibfnamefont {S.}~\bibnamefont {Boixo}}, \bibinfo {author} {\bibfnamefont {S.~T.}\ \bibnamefont {Flammia}}, \bibinfo {author} {\bibfnamefont {C.~M.}\ \bibnamefont {Caves}}, \ and\ \bibinfo {author} {\bibfnamefont {J.~M.}\ \bibnamefont {Geremia}},\ }\href {https://doi.org/10.1103/PhysRevLett.98.090401} {\bibfield  {journal} {\bibinfo  {journal} {Phys. Rev. Lett.}\ }\textbf {\bibinfo {volume} {98}},\ \bibinfo {pages} {090401} (\bibinfo {year} {2007})}\BibitemShut {NoStop}%
\end{thebibliography}%

\onecolumngrid

\newpage
\begin{center}
\textbf{Supplemental Material} 
\end{center}
\setcounter{secnumdepth}{1}
\renewcommand{\thesection}{S\arabic{section}}
\setcounter{theorem}{0}
\renewcommand{\thetheorem}{S\arabic{theorem}}
\setcounter{lemma}{0}
\renewcommand{\thelemma}{S\arabic{lemma}}
\setcounter{equation}{0}
\renewcommand{\theequation}{S\arabic{equation}}
\setcounter{table}{0}
\renewcommand{\thetable}{S\arabic{table}}
\setcounter{figure}{0}
\renewcommand{\thefigure}{S\arabic{figure}}

This supplemental material includes the following: a summary of key definitions and a detailed definition of the family $\mathbb{F}_\beta$ (Section S1), a description of how the properties of $f\in\mathbb{F}$ imply properties of the means $m^f(x,y)$ (Section S2), a proof of the coherent-incoherent bound and the implications for a qubit (Section S3), a demonstration that the generalized bounds can yield arbitrarily large multiplicative improvements (Section S4), experimental details (Section S5), an example showing that all $f_\beta\in\mathbb{F}_\beta$ can yield the tightest bound (Section S6), details on the saturation of bounds (Section S7), an example of ``fast'' Hamiltonians (Section S8) and details on the bounds in terms of energy variance (Section S9).

\section{Summary of Key Definitions}\label{app:definitions}

In this section, we provide, in \cref{tab:my_label}, a summary of the key quantities and their mathematical definitions. In these definitions, as in the main text, $A$ denotes a general Hermitian observable.

\begin{table}[ht]
    \centering
    \begin{tabular}{|c|c|}
    \hline
        Quantity & Mathematical Expression  \\\hline
         $\dot a$ & $\Tr(\dot\rho A)$ \\
         $(\Delta A)^2$ & $\Tr(\rho A^2)-\Tr(\rho A)^2$ \\
         $A_0$ & $A-\Tr(\rho A)$\\
         $m^f(x,y)$ & $xf(x^{-1}y)$\\
         $^*(\Delta^f A)^2$ & $\Tr\left[A_0\, m^f(L_\rho, R_\rho) (A_0)\right]$ \\
         $\mathcal{L}^f$ & $m^f(L_\rho, R_\rho)^{-1}(\dot\rho)$\\
         $\mathcal{I}^f$ & $\Tr[\dot\rho m^f(L_\rho, R_\rho)^{-1}(\dot\rho)]$\\
         & $\qquad\quad=\Tr[\mathcal{L}^f m^f(L_\rho, R_\rho)(\mathcal{L}^f)]$\\
         \hline
    \end{tabular}

    *equals $(\Delta A)^2$ for $f=f_\sld$.
    \caption{Definitions of key mathematical quantities.}
    \label{tab:my_label}
\end{table}

We also provide a complete mathematical description of the one-parameter family of operator monotone functions $\mathbb{F}_\beta$ used in the main text. In particular,
\begin{equation}\label{eq:fbeta}
f_\beta(x):=\begin{cases}
\frac{\beta(1-\beta)(x-1)^2}{(x^\beta-1)(x^{1-\beta}-1)}, & \beta\in\big[-1,\frac{1}{2}\big)\setminus \{0\}\\
\frac{x-1}{\log x}, & \beta = 0 \\
\left(\frac{1+x^\beta}{2}\right)^{1/\beta}, & \beta\in\big[\frac{1}{2},1\big].
\end{cases}
\end{equation}

Some important examples from this family and their associated means are summarized in \cref{tab:means}.

\begin{table}[ht]
    \centering
    \begin{tabular}{|c|c|c|c|c|}
    \hline
        Name & $\beta$ & $f(x)$ & $m^f(x,y)$ \\\hline
        SLD & $1$ & $\frac{1+x}{2}$ & $\frac{x+y}{2}$\\
        Wigner-Yanase & $\frac{1}{2}$ & $\frac{1+\sqrt{x}}{4}$ & $\frac{\sqrt{x}+\sqrt{y}}{4}$\\
        RLD & $-1$ & $\frac{2x}{1+x}$ & $\frac{2xy}{x+y}$\\
         \hline
    \end{tabular}
    \caption{Important examples of $f\in\mathbb{F}_\beta$ and the associated means.}
    \label{tab:means}
\end{table}

 \section{From Operator Monotone Functions to Means}\label{app:means}
 In this section, we review how the properties of the operator monotone functions $\mathbb{F}$ imply certain natural properties of the associated means. Recall that $m^f(x,y):=xf(x^{-1}y)$. Here, we consider $x,y\in\mathbb{R}^+$ or as commuting, positive definite (super-) operators. As described in the main text, the operator monotone functions $f\in\mathbb{F}$ have the following properties: 
 \begin{enumerate}[(i)]
 \item Continuity.
 \item For matrices $A\geq B> 0$, $f(A)\geq f(B)> 0$ (operator monotonicity). 
 \item $x f(x^{-1})=f(x)$, for $x> 0$ (symmetry).
 \item $f(1)=1$ (normalization).
 \end{enumerate}

 These properties of $f\in\mathbb{F}$ imply the following (non-exhaustive) list of properties for the associated means:
  \begin{enumerate}[(i')]
\item $m^f(x,y)$ is continuous.

Follows directly from property (i) and the definition of $m^f(x,y)$.
  
 \item $m^f(x,x)=x$.

\noindent Follows from property (iv):
\[
m^f(x,x)=xf(x^{-1}x)=xf(1)=x.
\]

 \item $m^f(x,y)=m^f(y,x)$.

\noindent Follows from property (iii): 
 \begin{align*}
     m^f(x,y)&=\frac{x}{xy^{-1}} f(y^{-1}x)\\
     &=yf(y^{-1}x)=m^f(y,x).
 \end{align*}

 \item $x\leq y$ $\implies $ $x\leq m^f(x,y)\leq y$.

\noindent Follows from properties (ii), (iv) and (iii'): 
 \begin{align*}
     m^f(x,y)&=x f(x^{-1}y)\geq x f(1)=x,\\
     m^f(x,y)&=y f(y^{-1}x)\leq y f(1)=y.
 \end{align*}

  \item $x\leq x'$, $y\leq y'$ $\implies $ $m^f(x,y)\leq m^f(x',y')$.

\noindent Follows from properties (ii), (iv), and (iii') almost identically to the demonstration of property (iv').
 \end{enumerate}

\section{Coherent-Incoherent Bound}\label{app:cross-terms}
In this section, we derive the split coherent-incoherent bound of \cref{eq:new-upper}, repeated here for convenience:
\begin{equation}\label{eq:new-upper-app}
|\dot a|\leq (\Delta^f A_C)\sqrt{\mathcal{I}_C^f}+(\Delta A_I) \sqrt{\mathcal{I}_I}.
\end{equation}

We also demonstrate how to express the quantities in this bound in the eigenbasis of $\rho=\sum_j p_j\ket{j}\bra{j}$ and use these expressions to show that the split coherent-incoherent bound in \cref{eq:new-upper-app} is generically tighter than the non-split bound in \cref{eq:new-bnd} of the main text.

\subsection{Proving the Coherent-Incoherent Bound}

To derive \cref{eq:new-upper-app}, note that 
\begin{align}\label{eq:split}
\dot a&=\Tr[m^f(L_\rho,R_\rho)(\mathcal{L}_I^f+\mathcal{L}^f_C)(A_{0,I}+A_{0,C})] \nonumber \\
&= \Tr[m^f(L_\rho,R_\rho)(\mathcal{L}_I^f)A_{0,I}] \nonumber\\
&\quad + \Tr[m^f(L_\rho,R_\rho)(\mathcal{L}^f_C)A_{0,C}],
\end{align}
where we use that the cross terms vanish under the trace. The subscripts $C$ and $I$ denote the coherent and incoherent parts of each operator in the eigenbasis of $\rho$, respectively. The bound in \cref{eq:new-upper-app} then follows immediately by applying the Cauchy-Schwarz inequality.

Therefore, it remains to explicitly show that the cross terms indeed vanish. That is, we must show that
\begin{align}\label{eq:cross-term-1}
\Tr[m^f(L_\rho,R_\rho)(\mathcal{L}^f_I)A_{0,C}]&=0, 
\\
\Tr[m^f(L_\rho,R_\rho)(\mathcal{L}^f_C)A_{0,I}]&=0. \label{eq:cross-term-2}
\end{align}
To show this we will use the matrix representation of the super-operators $L_\rho=\rho\otimes I$ and $R_\rho=I\otimes \rho^T$, which act on vectorized versions of the regular operators. That is, for an operator $O$, we denote its vectorized form in an arbitrary basis as
\begin{equation}
\oket{O}:=\sum_{i,j}O_{ij}\ket{i}\otimes\ket{j}=\sum_{i,j}O_{ij}\oket{ij},
\end{equation}
where we use the rounded ``ket'' $|\cdot )$ to denote vectors in the super-operator vector space.

Specifically, in the eigenbasis of $\rho=\sum_j p_j\ket{j}\bra{j}$, we can write $L_\rho=\sum_{ij}p_i\oket{ij}\obra{ij}$ and $R_\rho=\sum_{ij}p_j\oket{ij}\obra{ij}$. It is then easy to show, in this basis, that
\begin{align}\label{eq:mean-in-rho-basis}
m^f(L_\rho, R_\rho)&:=L_\rho f(L_\rho^{-1}R_\rho)\nonumber\\
&= \sum_{ij} p_if(p_i^{-1}p_j)\oket{ij}\obra{ij}\nonumber \\
&=\sum_{ij}m^f(p_i,p_j)\oket{ij}\obra{ij},
\end{align}
where we used that $m^f(p_i,p_j):=p_i f(p_i^{-1}p_j)$. Consequently, 
\begin{align}
&m^f(L_\rho, R_\rho)(\mathcal{L}_I^f)\nonumber\\
&\quad=\sum_{ij}m^f(p_i,p_j)\oket{ij}\obra{ij}\sum_{k}(\mathcal{L}_I^f)_{kk}\oket{kk}\nonumber\\
&\quad=\sum_{i}(\mathcal{L}_I^f)_{ii} m^f(p_i,p_i)\oket{ii}\nonumber\\
&\quad =\sum_{i}p_i(\mathcal{L}_I^f)_{ii} \oket{ii},
\end{align}
from which it it is easy to show \cref{eq:cross-term-1}, as $A_{0,C}$ is off-diagonal in this basis. Similar computations allow one to show \cref{eq:cross-term-2}.

\subsection{Generalized Quantities in State Eigenbasis}
For analysis of the bounds it is helpful to derive expressions for the elements of the split coherent-incoherent bound in \cref{eq:new-upper} (alternatively, \cref{eq:new-upper-app})  in the eigenbasis of $\rho=\sum_j p_j \ket{j}\bra{j}$. In particular, it holds that: 
\begin{align}
(\Delta^fA_C)^2&=\sum_{i\neq j} \big|(A_0)_{ij}\big|^2m^f(p_i,p_j),
\label{eq:var-in-rho-basis-C} \\
\mathcal{I}^f_C:=\left(\Delta\mathcal{L}^f_C\right)^2&=\sum_{i\neq j} \big|(\dot \rho)_{ij}\big|^2\frac{1}{m^f(p_i,p_j)},\label{eq:I-in-rho-basis-C}\\
(\Delta A_I)^2&=\sum_{i} p_i\big|(A_0)_{ii}\big|^2,
\label{eq:var-in-rho-basis-I}\\
\mathcal{I}_I:=\left(\Delta\mathcal{L}^f_I\right)^2&=\sum_j p_j \left(\frac{d \log p_j}{dt}\right)^2.\label{eq:fi}
\end{align}
Recall, we drop the $f$-superscript on the incoherent terms because, for incoherent operators, the generalized quantum Fisher information and variance are identical for all $f$~\footnote{Just use that $m^f(p_i,p_i)=p_i$ for all $f$}. Furthermore, $\mathcal{I}_I$ is simply the classical Fisher information of the probability distribution specified by the eigenvalues $\{p_j\}_{j=0}^{d-1}$ of $\rho$.

We now sketch the derivation for these expressions. Using \cref{eq:mean-in-rho-basis},
%
\begin{align}
 &m^f(L_\rho, R_\rho) (A_0)\nonumber \\
 &\qquad=\sum_{ij}m^f(p_i,p_j)\oket{ij}\obra{ij} \sum_{kl}(A_0)_{kl}\oket{kl}\nonumber\\
&\qquad=\sum_{ij}(A_0)_{ij}m^f(p_i,p_j)\oket{ij}.
\end{align}
Undoing the vectorization and plugging into \cref{eq:gen-var} of the main text yields
\begin{align}\label{eq:var-in-rho-basis-app}
&(\Delta^fA)^2:=\Tr\left[A_0 m^f(L_\rho, R_\rho) (A_0)\right]\nonumber\\
&=\Tr \bigg[\sum_{ijkl}(A_0)_{kl}\ket{k}\bra{l}(A_0)_{ij}m^f(p_i,p_j)\ket{i}\bra{j}\bigg]\nonumber\\
&=\sum_{ij} \big|(A_0)_{ij}\big|^2m^f(p_i,p_j).
\end{align}
Restricting $A$ to its coherent component $A\rightarrow A_C$ yields \cref{eq:var-in-rho-basis-C}. Similarly, considering $A\rightarrow A_I$ and using that $m^f(p_i,p_i)=p_i$ yields \cref{eq:var-in-rho-basis-I}. Essentially identical manipulations prove \cref{eq:I-in-rho-basis-C,eq:fi} starting from the definition of $\mathcal{I}^f$ in \cref{eq:gen-qfi} and using that $\dot\rho=m^f(L_\rho,R_\rho)(\mathcal{L}^f)$.

\subsection{All bounds are equivalent for a qubit}
The results of the previous section can be used to demonstrate a simple proposition, the results of which were described in the main text.
\begin{proposition}
For a qubit, $(\Delta^f A_C)\sqrt{\mathcal{I}^f_C}=(\Delta^\sld A_C)\sqrt{\mathcal{I}^\sld_C}$ for all $f\in\mathbb{F}$. 
\end{proposition}
\begin{proof}
In both \cref{eq:var-in-rho-basis-C,eq:I-in-rho-basis-C}, 
there is a single pair of equal terms for a qubit. 
Thus, working in the eigenbasis of $\rho$, $(\Delta^f A_C)\sqrt{\mathcal{I}^f_C}=2|(A_0)_{01}||\dot \rho_{01}|$ for all $f$, proving the result.
\end{proof}

\subsection{Coherent-Incoherent Bound is a Generically Tighter Upper Bound}
As stated in the main text, the coherent-incoherent bound of \cref{eq:new-upper} of the main text (alternatively, \cref{eq:new-upper-app}) is always at least as tight as the non-split bound of \cref{eq:new-bnd} of the main text,
\begin{equation}\label{eq:new-bnd-app}
|\dot a|\leq (\Delta^f A)\sqrt{\mathcal{I}^f}.
\end{equation}

To prove this, we start from \cref{eq:var-in-rho-basis-app} for the generalized quantum $f$-variance of $A$ in the eigenbasis of the state $\rho$. We then split $A_0:=A-\mathrm{Tr}(\rho A)$ into its coherent and incoherent parts and simplify the resulting expression as follows:
\begin{align}
(\Delta^fA)^2&=\sum_{ij} \big|(A_0)_{ij}\big|^2m^f(p_i,p_j)\nonumber \\
&=\sum_{ij} \Big[\big|(A_{0,I})_{ij}\big|^2 + \big|(A_{0,C})_{ij}\big|^2\Big]m^f(p_i,p_j)\nonumber \\
&= (\Delta^f A_{I})^2+(\Delta^f A_{C})^2\nonumber\\
&=  (\Delta A_{I})^2+(\Delta^f A_{C})^2,
\end{align}
where in the last line we drop the $f$ subscript for the variance of the incoherent operator $A_I$, as for diagonal operators the variance evaluates to the usual variance for all $f$.

Therefore, the ratio of the right-hand side of \cref{eq:new-bnd-app} and the right-hand side of \cref{eq:new-upper-app} can be computed to be
\begin{align}
\frac{(\Delta^f A)\sqrt{\mathcal{I}^f}}{(\Delta^f A_C)\sqrt{\mathcal{I}_C^f}+(\Delta A_I) \sqrt{\mathcal{I}_I}}
&=\sqrt{\frac{\big[(\Delta A_I)^2+(\Delta^f A_C)^2\big]\big[\mathcal{I}^f_C+\mathcal{I}_I\big]}{\Big[(\Delta^f A_C)\sqrt{\mathcal{I}_C^f}+(\Delta A_I) \sqrt{\mathcal{I}_I}\Big]^2}}\nonumber \\
&=\sqrt{1+\frac{(\Delta^fA_C)^2\mathcal{I}_I+(\Delta A_I)^2\mathcal{I}^f_C-2(\Delta^fA_C)(\Delta A_I)\sqrt{\mathcal{I}^f_C\mathcal{I}_I}}{\Big[(\Delta^f A_C)\sqrt{\mathcal{I}_C^f}+(\Delta A_I) \sqrt{\mathcal{I}_I}\Big]^2}}\nonumber\\
&=\sqrt{1+\frac{\bigg[(\Delta^fA_C)\sqrt{\mathcal{I}_I}-(\Delta A_I)\sqrt{\mathcal{I}^f_C}\bigg]^2}{\Big[(\Delta^f A_C)\sqrt{\mathcal{I}_C^f}+(\Delta A_I) \sqrt{\mathcal{I}_I}\Big]^2}}\geq 1.
\end{align}
Thus, the coherent-incoherent bound in \cref{eq:new-upper-app} is tighter than the non-split bound in \cref{eq:new-bnd-app}. The two are equal if and only if $(\Delta^fA_C)\sqrt{\mathcal{I}_I}=(\Delta A_I)\sqrt{\mathcal{I}^f_C}$. This precisely mirrors the results derived in Appendix F of Ref.~\cite{lgp2022unifying} for the special case $f=f_\sld$.

\section{Arbitrarily Large Multiplicative Improvements from Generalized Bounds}\label{app:largest-improvement}
In this section, we demonstrate the claim that arbitrarily large multiplicative improvements in the bounds are possible compared to the existing $f_\sld$-based bounds of Ref.~\cite{lgp2022unifying}. In particular, we show this by example by demonstrating that the quantity $\xi^f$ can be made arbitrarily small in the qutrit example of the main text.

It is simple to compute that for this example $v_{01}=\frac{1}{2}\Omega^2 (p_1-p_0)^2$ and $v_{02}=\frac{1}{2}\Omega^2(p_2-p_1)^2$. Then using \cref{eq:qutrit-comp} for $f=f_\rld$, the fact that $m_{ij}^\rld=2p_ip_j/(p_i+p_j)$ and $m_{ij}^\sld = (p_i+p_j)/2$, and a bit of elementary algebra we have that
\begin{align}
&\xi^\rld
=\frac{p_1+p_2}{p_2(p_0+p_1)}\left[\frac{(p_1-p_0)^2(p_0+p_1)p_2+(p_2-p_1)^2p_0(p_1+p_2)}{(p_1-p_0)^2(p_1+p_2)+(p_2-p_1)^2(p_0+p_1)}\right].
\end{align}

Consider $p_0=\epsilon^3, p_1=\epsilon, p_2=1-\epsilon-\epsilon^3$ for some small $\epsilon$. Then,
\begin{equation}
\xi^\rld= \epsilon + \mathcal{O}(\epsilon^2),
\end{equation}
proving the desired result. Such states correspond to the lower right corner near $p_2=1$ in \cref{fig:example-1} of the main text.

\section{Experimental Details}\label{app:experiment}
In this section, we provide details behind the experimental demonstration discussed in the qutrit example described in the main text. 
We use a superconducting transmon qutrit \cite{Wallraff_qutrit} to collect the data presented in the proof-of-concept measurement in Fig. \ref{fig:example-1}. The qutrit is a standard X-mon style device with a Ta ground plane on a sapphire substrate and a Al/AlO$_\text{x}$/Al junction. The qutrit is coupled to a coplanar waveguide resonator for dispersive readout of the qutrit state  \cite{Blais2021_cQED, krantz2019A}. A dedicated drive line is used to apply coherent control on the qutrit (see Ref. \cite{Huang_packaging} for further details on the measurement setup). The relevant qutrit parameters are listed in \cref{tab:exp_params}.

We initialize the qutrit in the $\ket{f}$ state using two consecutive $\pi$-pulses on the $\ket{g}-\ket{e}$ and $\ket{e}-\ket{f}$ transitions. By varying the delay time between the state initialization and the beginning of the experiment, we utilize the natural depolarization of the excited states to generate purely diagonal mixed states in the $\ket{g},\ket{e}, \ket{f}$ basis. We fit the population dynamics to a simple decay chain model given by the Bateman equation \cite{bateman_eq, f_state_decay} to verify our initialization scheme and qutrit state readout as shown in \cref{fig:bateman-populations}. This data is used to determine population values and uncertainties used in \cref{fig:example-1}.

To carry out the experiment shown in \cref{fig:example-1}, we initialize the qutrit in a mixed state via the depolarization technique described above. After the state is initialized, we estimate $\dot{X}_{ge}$ for the specified Hamiltonian, $H=i(\Omega/2)(\ket{g}\bra{e}+\ket{e}\bra{f})+h.c.$, by applying simultaneous drives for a variable amount of time on the $\ket{g}-\ket{e}$ transition with drive strength $\Omega_{ge}$ and $\ket{e}-\ket{f}$ transition with drive strength $\Omega_{ef}$ about the Pauli-$Y$ axis of the Bloch spheres associated with these subspaces. We measure $\langle X_{ge} \rangle$ by applying a $\pi/2$-pulse before qutrit state readout to map the observable onto the population of the qutrit state: $\langle X_{ge} \rangle = P_e - P_g$, where $P_{g, (e)}$ is the extracted population of the $\ket{g}$ ($\ket{e}$) qutrit level. We measure the rate of change of the populations by varying the drive time and fitting a line to the short time dynamics ($0$~ns to $16$~ns evolution time) of $P_{g,e}$. We use the population rates of change
to calculate the speed $\dot{X}_{ge} = \dot{P}_e - \dot{P}_g$ under the measurement mapping. 

We collect data points with interleaved $\ket{g}, \ket{e}, \ket{f}$ reference measurements for population extraction. For each point, we collected 20000 averages. 
We compute uncertainties of fitted $\dot{P}_{g,e}$ from the linear least-squares fit of the populations. These uncertainties are used to obtain the the error bars on $\dot{X}_{ge}$ shown in \cref{fig:example-1} of the main text.
Data available upon 
request. 

\begin{figure}
    \centering
    \includegraphics[width=0.5\columnwidth]{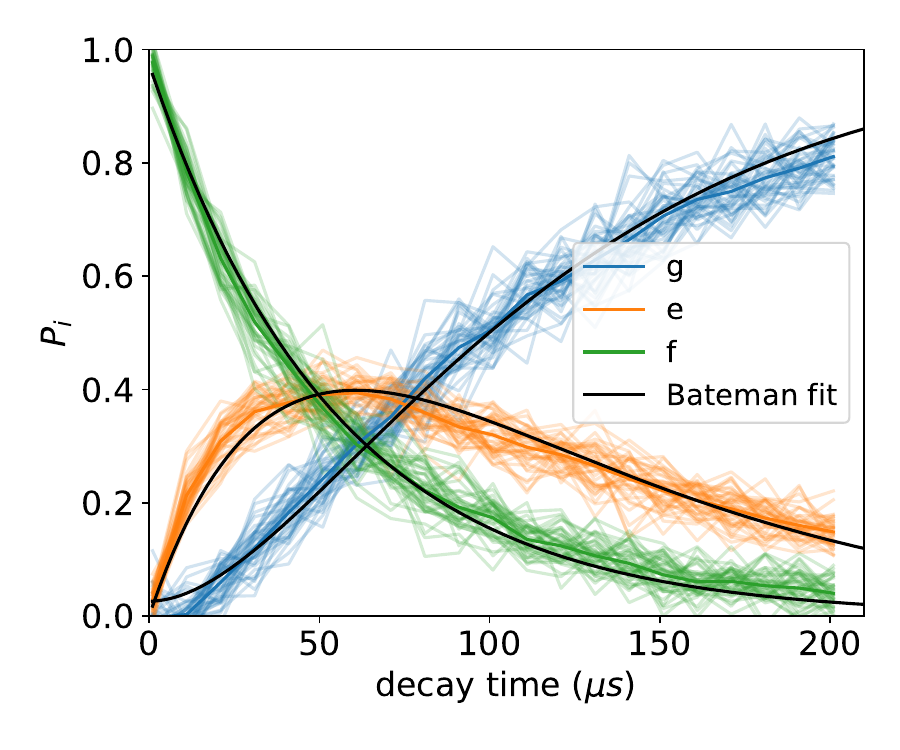}
    \caption{Repeated measurements of qutrit populations after variable decay time from the $\ket{f}$ state. Solid lines are the fit to the Bateman equations with decay rates listed in \cref{tab:exp_params}. Data is used to calibrate the state preparation technique used in \cref{fig:example-1}.
    }
    \label{fig:bateman-populations}
\end{figure}

\begin{table}[]
    \centering
    \begin{tabular}{c|c}
        Parameter & Value \\
        \hline
        $\omega_{ge}$ & $2\pi \cdot 2.78~\text{GHz}$ \\
        $\omega_{ef}$ & $2\pi \cdot 2.63~\text{GHz}$ \\
        $\Omega_{ge}$ & $2\pi \cdot 10~\text{MHz}$\\
        $\Omega_{ef}$ & $2\pi \cdot 10~\text{MHz}$ \\
        $\Gamma_{eg}$ & $14.7~\mu \text{s}^{-1}$ \\
        $\Gamma_{fe}$ & $18.4~\mu \text{s}^{-1}$ \\
    \end{tabular}
    \caption{Qutrit parameters for the device used in the experimental measurements. $\nu_{ge (ef)}$ is the transition frequency between the $\ket{g}$ and $\ket{e}$ ($\ket{e}$ and $\ket{f}$) states of the qutrit. $\Gamma_{eg (fe)}$ is the measured decay rate from the $\ket{e}$ to $\ket{g}$ ($\ket{f}$ to $\ket{e}$) state (defined as 1/$T_1$). $\Omega_{ge (ef)}$ is the Rabi rate applied to the $\ket{g}$ to $\ket{e}$ ($\ket{e}$ to $\ket{f}$) transition.}
    \label{tab:exp_params}
\end{table}

\section{Any Bound Can be Tightest}\label{app:example-details}
In this section, we provide an example demonstrating the claim in the main text that for every $f_{\beta}\in\mathbb{F}_\beta$, there exists $(\rho,\dot\rho,A)$ such that the tightest bound, optimized over this family, is $f_{\beta}$. This example is similar to the qutrit example of the main text. Let the observable of interest is $A=\ket{g}\bra{e}+\ket{e}\bra{g}$, and, as in that example, we focus on diagonal $\rho$, so that this corresponds to $A=\ket{0}\bra{1}+\ket{1}\bra{0}$.

Consider driving the system with the Hamiltonian $H=-(i 2\pi)(2.5 \,\mathrm{ MHz})(\ket{0}\bra{1}+4\ket{1}\bra{2}+\ket{0}\bra{2})+h.c.$. In \cref{fig:example-2}, we show the $\beta$ that minimizes the coherent ratio $\xi^{f_\beta}$ as a function of the eigenvalues $\{p_j\}_{j=0}^2$ of $\rho$, revealing a rich structure in the $f_\beta$ yielding the tightest bound.

\begin{figure}[hb]
\centering
\includegraphics[width=0.5\columnwidth]{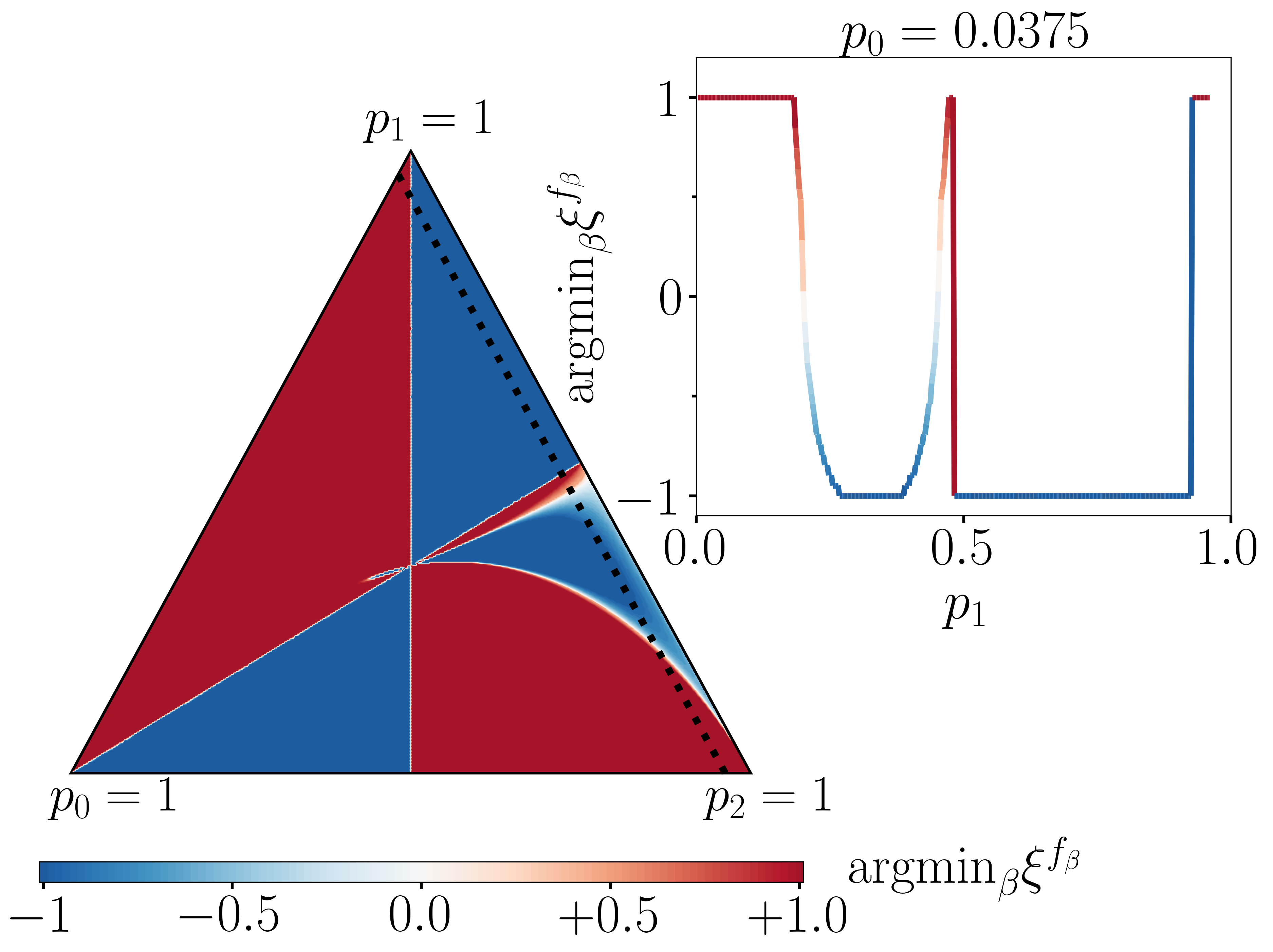}
\caption{For a qutrit, optimal $f_\beta\in\mathbb{F}_\beta$ to minimize the coherent ratio $\xi^{f_\beta}$ (\cref{eq:qutrit-comp} with $v_{01},v_{02},v_{12}\neq 0$) as a function of the eigenvalues $\{p_j\}_{j=0}^2$ of $\rho$. In the inset, the optimal $\beta$ is shown as a function of $p_1$ for fixed $p_0=0.0375$ (along dotted line in main figure).}\label{fig:example-2}
\end{figure}

\section{Saturation of Bounds}\label{app:saturation}
In this section, we provide additional details on when the bounds considered in this work are saturated. 

Starting with \cref{eq:new-bnd} of the main text, this speed limit is tight when $A_0\propto \mathcal{L}^f$. To better understand this condition, it is, again, helpful to work in the eigenbasis of $\rho$. In the super-operator vector space,
\begin{align}
\mathcal{L}^f&:=m^f(L_\rho, R_\rho)^{-1}(\dot\rho)\nonumber\\
&=\sum_{ij}\frac{1}{m^f(p_i,p_j)}\oket{ij}\obra{ij}\sum_{kl}\dot\rho_{kl}\oket{kl}\nonumber\\
&=\sum_{ij}\frac{\dot\rho_{ij}}{m^f(p_i,p_j)}\oket{ij}.
\end{align}
So \cref{eq:new-bnd} is tight if (and only if)
\begin{equation}\label{eq:A-sat-cond}
(A_0)_{ij}=\gamma\frac{\dot\rho_{ij}}{m^f(p_i,p_j)},
\end{equation}
for all $i,j$ and some constant $\gamma$. Observe the $f$-dependence of this saturation condition. That is, given some fixed observable $A$, dynamics, and a state $\rho$, the bounds are generically saturable only for a particular choice of $f$. There are special cases where the $f$-dependence is eliminated: for instance, if everything is fully incoherent, the fact that $m^f(p_i,p_i)=p_i$ for all $f$ removes the $f$ dependence. 

Clearly, if and only if $A_0\propto \mathcal{L}^f$, then also $A_C\propto\mathcal{L}^f_C$ and $A_I\propto \mathcal{L}_I$ (where we drop the $f$ subscript because it is irrelevant for incoherent terms). So, this condition also implies the saturation of the split coherent-incoherent upper bound in \cref{eq:new-upper}.

Note that, for the example in the main text (a qutrit with a fully coherent observable and only $a_{01}\neq 0$), if the bounds are saturated for any $f$, they are saturated for all $f$. In particular, $|\dot{a}|=(\Delta^f A)\sqrt{\mathcal{I}^f}=(\Delta^f A_C)\sqrt{\mathcal{I}^f_C}$ for all $f$ in this example if and only if $v_{01}\neq 0$ and $v_{12}=v_{02}=0$. This is because only $A_{01}, A_{10}$ are non-zero for this choice of observable; thus, the saturability condition in \cref{eq:A-sat-cond} is especially simple, and the $f$ dependence can be absorbed into the constant $\gamma$. 
If there is more than a single relevant pair $i,j$ in the saturability condition of \cref{eq:A-sat-cond}, then we can recover the $f$-dependence.

For the case of purely coherent dynamics driven by a Hamiltonian $H$, it is informative to rewrite \cref{eq:A-sat-cond} as
\begin{equation}\label{eq:sat-con-coherent}
A_{ij}=-i\gamma\frac{H_{ij}(p_j-p_i)}{m^f(p_i,p_j)},
\end{equation}
where we use that $\dot\rho=-i[H,\rho]$ and $H_{ij}$ are the matrix elements of $H$ in the eigenbasis of $\rho$.

\subsection{Example of Saturating the Bounds}
In the main text, we claim one can construct examples where the SLD-based bounds are loose, but one of our new generalized bounds can be tight for the appropriate choice of $f$. Here, we show such an example.

To begin, as $\gamma$ is independent of $i$ and $j$, \cref{eq:sat-con-coherent} implies the necessary condition for the saturation of the bound in \cref{eq:new-upper} for purely coherent dynamics:
\begin{equation}
\frac{H_{ij}(p_j-p_i)}{A_{ij}m^f(p_i,p_j)}=\frac{H_{kl}(p_l-p_k)}{A_{kl}m^f(p_k,p_l)}
\end{equation}
for all $i,j,k,l$ such that $A_{ij}, H_{ij}, A_{kl}, H_{kl} \neq 0$. This condition is not sufficient as \cref{eq:sat-con-coherent} must also hold for matrix elements where $A_{ij}=0$ and/or $H_{ij}=0$. Let us assume for simplicity that these conditions are satisfied trivially so that, when $H_{ij}=0$, also $A_{ij}=0$. 

Define $x:=p_i/p_j$, $y:=p_k/p_l$, and $c=H_{kl}A_{ij}/(H_{ij}A_{kl})$. Then, we have the necessary condition(s):
\begin{equation}\label{eq:b}
    cf(x)-f(y)=cyf(x)-xf(y),
\end{equation}
where $x,y>0$. Generically, $c\in\mathbb{C}$, but it is clear (as $x,y,f(x),f(y)>0$) that, unless $c\in\mathbb{R}$, \cref{eq:b} has no solutions. So we restrict our attention to the case that $c\in\mathbb{R}$.

Consider two qubits with $H=X_1 X_2$ and  let $A=Y_1X_2-X_1Y_2/2$. We have one non-trivial saturation condition in this example corresponding to an identification of the indices $i\rightarrow 00, j\rightarrow 11, k\rightarrow 01,$ and $l\rightarrow 10$. The associated $c=\frac{1}{3}$. Also, $x:=p_{00}/p_{11}$ and $y:=p_{01}/p_{10}$.

It is fairly straightforward to work out that, in this setting, with $f(t)=f_\sld(t)=(1+t)/2$, \cref{eq:b} becomes
\begin{equation}\label{eq:sldexample}
x=\frac{1+2y}{2+y}.
\end{equation}
As a concrete example, $p_{00}=0.3, p_{01}=0.4, p_{10}=0.1, p_{11}=0.2$, corresponding to $y=4$ and $x=3/2$, obeys this saturation condition.

In contrast to the first example, valid solutions to the condition for $f(t)=f_\rld(t)=2t/(1+t)$ do not correspond to those for the condition in \cref{eq:sldexample}. In particular, here, \cref{eq:b} becomes
\begin{equation}\label{eq:rldexample}
x=\frac{y^2-1\pm \sqrt{34y^2+y^4+1}}{6y}.
\end{equation}
Any solutions to this equation with $x,y>0$ correspond to states that satisfy the generalized quantum speed limit associated with $f_\rld$. For instance, $p_{00}=( \sqrt{89}-3)/20, p_{01}=0.4, p_{10}=0.1, p_{11}=(13 - \sqrt{89})/20$,  corresponding to $y=4, x=(15+3\sqrt{89})/24$, satisfies \cref{eq:rldexample}, but clearly not \cref{eq:sldexample}. Therefore, we have constructed an example where the generalized bounds can be tight when the SLD-based bounds are not.

\section{Fast Hamiltonian Example}\label{app:fastH}
In this section, we provide an example where the optimal ``fast'' driving Hamiltonians (see \cref{eq:Hfast} of the main text) derived based on the generalized quantum speed limits can yield larger $|\dot a|$ than the existing SLD-based fast driving Hamiltonians.

In particular, consider a qutrit with the observable of interest as $A=\ket{0}\bra{1}+\ket{1}\bra{2}+h.c.$. For this observable and a range of different states $\rho$, we optimize over the ``fast'' Hamiltonians $H^{f,\mathrm{fast}}_{\rho,A}$ with respect to functions $f\in\mathbb{F}_\beta$ (see \cref{eq:fbeta}) for fixed spectral norm. In \cref{fig:example-fast-H}, we plot the ratio of the speed with the optimal $H^{f,\mathrm{fast}}_{\rho,A}$ from the result of this maximization to the speed with $H^{f_\sld,\mathrm{fast}}_{\rho,A}$ as a function of the eigenvalues $p_0,p_1$ of $\rho$. Observe that there is a whole band of states where the optimal driving Hamiltonian is derived from $f\neq f_\sld$. Therefore, we have a demonstration that the generalized bounds can be used to derive faster driving Hamiltonians than the usual SLD-based bounds. Here, the driving Hamiltonians based on the generalized bounds can increase $|\dot a|$ by up to approximately twenty percent relative to the SLD-based drives.

\begin{figure}[h]
\centering
\includegraphics[width=0.5\columnwidth]{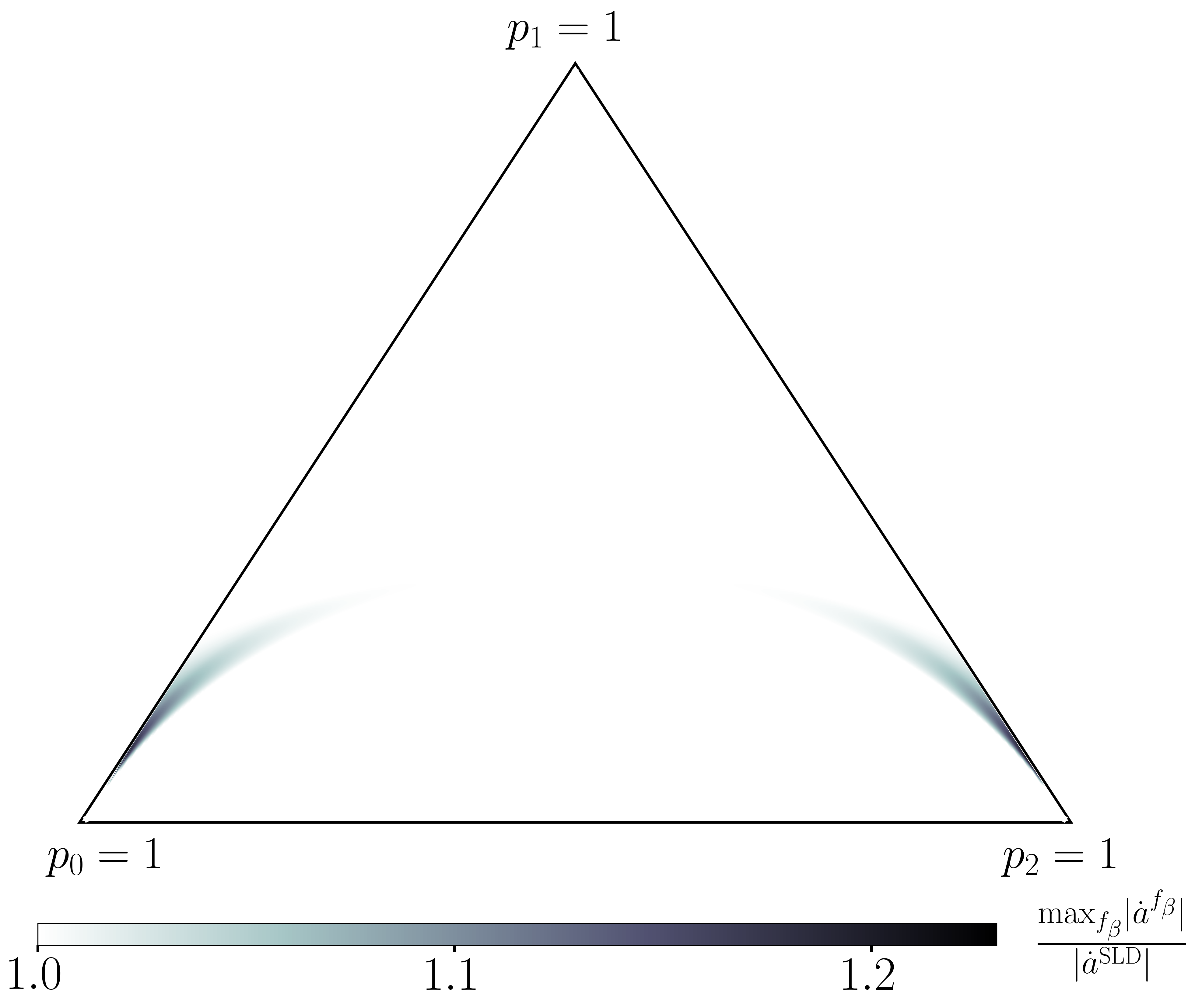}
\caption{Plot of the ratio of the speed with the optimal $H^{f,\mathrm{fast}}_{\rho,A}$ to the speed with $H^{f_\sld,\mathrm{fast}}_{\rho,A}$ as a function of the eigenvalues $p_0,p_1$ of $\rho$. Note the bands of states where $f\neq f_\sld$ yield larger $|\dot a|$. }\label{fig:example-fast-H}
\end{figure}

\section{Bounds in Terms of Energy Variance}\label{app:energy-bnds}
In this section, we show that at the cost of loosening the bounds we can obtain bounds similar to \cref{eq:new-upper} that depend on physical quantities more easily accessible and easily interpretable than the quantum Fisher informations $\mathcal{I}^f$. 

In particular, we can show that
\begin{align}\label{eq:energy-bnd}
\mathcal{I}^f_C&\leq 4\max_{p_i,p_j}\left(\frac{m^\sld(p_i,p_j)}{m^f(p_i,p_j)}\right)(\Delta H)^2 
\end{align}
where $(\Delta H)^2$ is the (usual, SLD-based) variance of the Hamiltonian $H$ driving the coherent dynamics. In the case of $f=f_\sld$ this reduces to the well-known relationship $\mathcal{I}^\sld_C\leq 4\Delta H$~\cite{Liu_2019}. 

To see \cref{eq:energy-bnd} we use $\dot\rho=-i[\rho,H]$ and \cref{eq:I-in-rho-basis-C} to write
\begin{align}
\mathcal{I}^f_C &= \sum_{i\neq j}\frac{(p_i-p_j)^2}{m^f(p_i,p_j)}|H_{ij}|^2 \nonumber \\
&\leq \sum_{i\neq j}\frac{(p_i+p_j)^2}{m^f(p_i,p_j)}|H_{ij}|^2 .
\end{align}
Then, from the definition of $m^\sld(p_i,p_j)$, which corresponds to the arithmetic mean,  $(p_i+p_j)=2m^\sld(p_i,p_j)$, so
\begin{align}
\mathcal{I}^f_C &\leq 4 \sum_{i\neq j}\frac{(m^\sld(p_i,p_j))^2}{m^f(p_i,p_j)}|H_{ij}|^2 \nonumber \\
&\leq 4 \max_{p_i,p_j}\left(\frac{m^\sld(p_i,p_j)}{m^f(p_i,p_j)}\right)\sum_{i\neq j}m^\sld(p_i,p_j)|H_{ij}|^2 \nonumber \\
&= 4 \max_{p_i,p_j}\left(\frac{m^\sld(p_i,p_j)}{m^f(p_i,p_j)}\right) (\Delta H_C)^2\nonumber \\
&\leq 4 \max_{p_i,p_j}\left(\frac{m^\sld(p_i,p_j)}{m^f(p_i,p_j)}\right)(\Delta H)^2, 
\end{align}
proving \cref{eq:energy-bnd}. In the next-to-last line we use \cref{eq:var-in-rho-basis-C}, and in the last line we use that $(\Delta H_C)^2 \leq (\Delta H)^2 $.

Even simpler bounds can be found by considering $f=f_\rld$ and using that $m^f(p_i,p_j)\geq m^\rld(p_i,p_j)$ for all $f\in\mathbb{F}$:
\begin{equation}
\mathcal{I}^f_C\leq \max_{p_i,p_j}\left(\frac{(p_i+p_j)^2}{2p_ip_j}\right)(\Delta H)^2.
\end{equation}
Now, letting $\kappa_{ij}:=p_j/p_i$,
\begin{align}
\max_{p_i,p_j}\left(\frac{(p_i+p_j)^2}{2p_ip_j}\right) \mapsto \max_{\kappa_{ij}}\frac{1}{\kappa_{ij}}(1+\kappa_{ij}^2).
\end{align}
This expression is maximized by either making $\kappa_{ij}$ as large or small as possible (note that the expression above is identical with $\kappa_{ij}\rightarrow \kappa_{ij}^{-1}$). Without loss of generality, let us take $p_j\geq p_i$. Then the maximum $\kappa_{ij}=\max_{p_i,p_j}(p_j/p_i)$ is the condition number $\kappa_\rho$ of $\rho$. Therefore,
\begin{equation}\label{eq:energy-variance-final-bnd-C}
\mathcal{I}^f_C\leq \kappa_\rho (\Delta H)^2.
\end{equation}
As a final simplification (and further loosening of the bound), we can use that $4(\Delta H)^2\leq \norm{H}_s^2$ where the seminorm is defined as the difference in the maximum and minimum eigenvalues of $H$~\cite{boixo2008generalized}. This gives us a bound that depends only on the maximum and minimum eigenvalues of $\rho$ and $H$:
\begin{equation}\label{eq:seminorm-final-bnd-C}
\mathcal{I}^f_C\leq \kappa_\rho \norm{H}^2_s.
\end{equation}

Finally, if there are incoherent dynamics and we assume the source of this non-unitary dynamics is entanglement with an environment via a Hamiltonian $H^\mathrm{int}$ that includes all terms with support on both the system and environment then, as proven in Ref.~\cite{lgp2022unifying},
\begin{equation}\label{eq:energy-variance-I}
\mathcal{I}_I\leq 4(\Delta H^\mathrm{int})^2 \leq \norm{H^\mathrm{int}}_s^2
\end{equation}
where this variance is calculated on the joint state of the system and the environment. 

Putting together \cref{eq:energy-variance-final-bnd-C,eq:seminorm-final-bnd-C,eq:energy-variance-I} gives the quantum speed limit:
\begin{equation}\label{eq:final-variance-bnd}
|\dot a|\leq \sqrt{\kappa_\rho}\Delta^f A_C\Delta H+2\Delta A_I\Delta H^\mathrm{int}.
\end{equation}
As $\Delta^f A \leq \Delta A$ for all $f$, the tightest bound of this form  will always correspond to $f=f_\rld$. Even for $f=f_\sld$, however, \cref{eq:final-variance-bnd} is guaranteed to be tighter than the existing bounds~\cite{lgp2022unifying}
\begin{equation}\label{eq:lgp-variance-bnd}
|\dot a|\leq  2\Delta A_C\Delta H+2\Delta A_I\Delta H^\mathrm{int}
\end{equation}
if $\kappa_\rho<4$. For these states the new bound in \cref{eq:energy-variance-final-bnd-C} is tighter than the bound $\mathcal{I}_C^\sld < 4 (\Delta H)^2$ used in Ref.~\cite{lgp2022unifying}.
As described in the main text, the $\kappa_\rho<4$ condition is met, for instance, for a thermal state $\rho\propto \exp(-\beta H)$ with $\beta \leq \log(4)/(2\norm{H}_s)$.

However, as stated in the main text, the range of states where \cref{eq:final-variance-bnd} is tighter than \cref{eq:lgp-variance-bnd} will be much larger (and observable dependent) than the condition $\kappa_\rho<4$ suggests. This is because the regions of parameter space where $\kappa_\rho$ correspond closely with those where $\Delta^f A_C \ll \Delta A_C$, as can be seen from \cref{eq:var-in-rho-basis-C}.

For instance, consider the example in the main text. Recall in this example we consider a qutrit with coherent dynamics driven by a Hamiltonian $\propto (i\ket{0}\bra{1}+i\ket{1}\bra{2})+h.c.$ and the observable of interest is $A=A_C=\ket{0}\bra{1}+\ket{1}\bra{0}$. In \cref{fig:example-1-loose-bnds} we plot the ratio of the new bound in \cref{eq:final-variance-bnd} (\cref{eq:variance-bnd} of the main text) with $f=f_\rld$ to the bound in \cref{eq:lgp-variance-bnd} as a function of the eigenvalues $p_0,p_1,p_2$ of $\rho$. As expected, the new bound is tighter when $\kappa_\rho<4$, but there is also an extended observable-dependent region of parameter space beyond this region where the new bound is tighter.

\begin{figure}[ht]
\centering
\includegraphics[width=0.5\columnwidth]{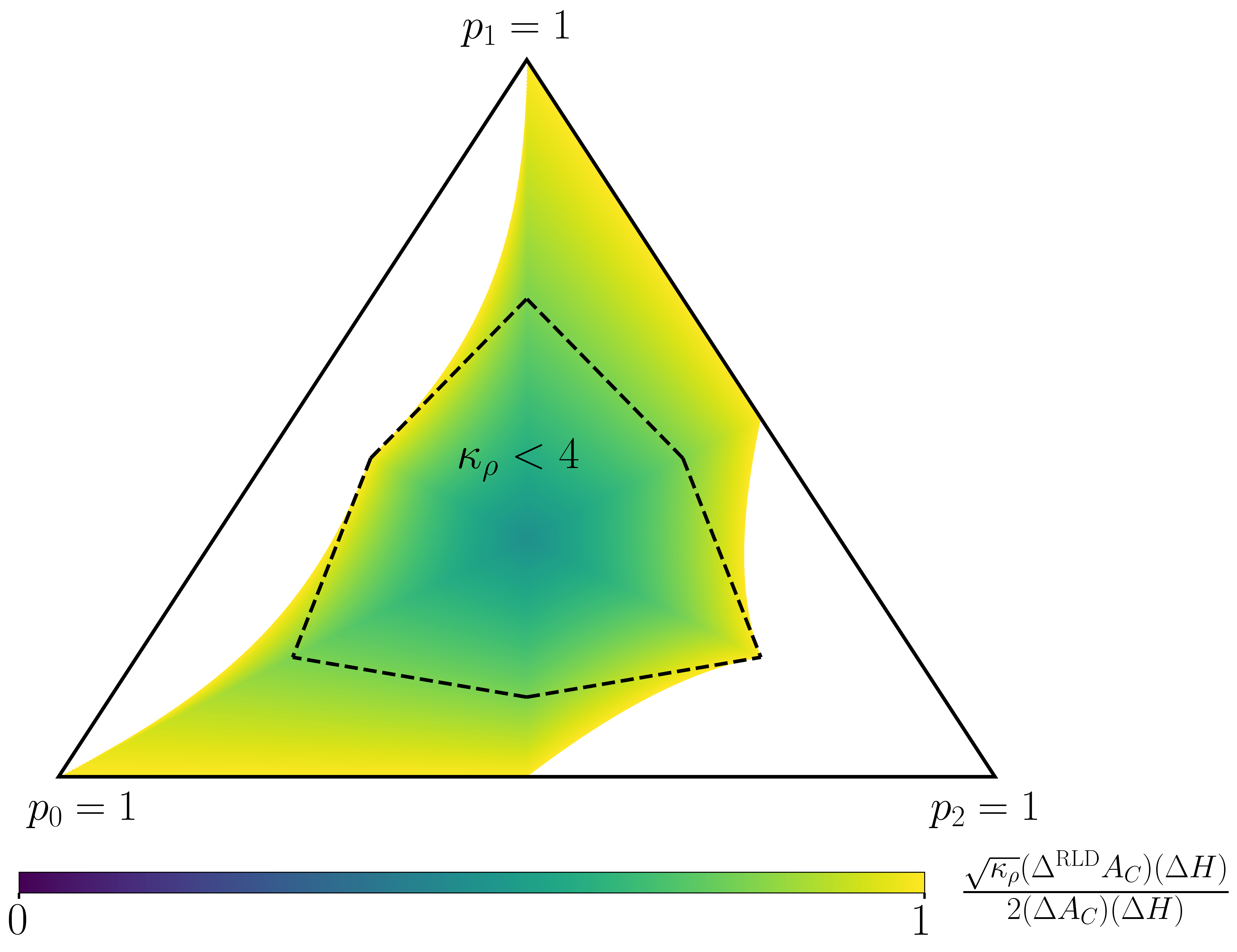}
\caption{For the qutrit example of the main text, plot of the ratio of the new bound in \cref{eq:final-variance-bnd} (\cref{eq:variance-bnd} of the main text) with $f=f_\rld$ to the bound in \cref{eq:lgp-variance-bnd} as a function of the eigenvalues $p_0,p_1,p_2$ of $\rho$, plotted in barycentric coordinates. The region where $\kappa<4$ is marked by the dashed lines. In this region, the new bound is guaranteed to be tighter for any observable $A$. However, there is also a large region of parameter space beyond this region where the new bound is still tighter.}\label{fig:example-1-loose-bnds}
\end{figure}
\end{document}